\documentclass[12pt,letterpaper]{JHEP3}

\usepackage{amscd,amsmath,amssymb,amsfonts,xspace,mathrsfs,amsthm,upgreek}
\usepackage{tikz,makecell}
\usetikzlibrary{decorations.markings,trees}
\usepackage{bbm}

\usepackage{color,mathtools}
\usepackage{url}
\let\normalcolor\relax
\usepackage{latexsym}
\usepackage{qtree}
\usepackage{graphicx}
\usepackage{dsfont}
\usepackage{color}
\usepackage{float}
\usepackage{longtable}

\usepackage{epsf}
\usepackage[latin1]{inputenc}
\usepackage{multirow}
\usepackage{epsfig}

\newtheorem{proposition}{Proposition}%[section]
\newtheorem{conj}{Conjecture}

\numberwithin{equation}{section}

\def\lfig#1#2#3#4#5{
\begin{figure}[t]
 \centerline{\includegraphics[width=#3]{#2}}
 \vspace{#5}
  \caption{#1 \label{#4}}
 \end{figure}
}

\def\bea{\begin{eqnarray}}
\def\eea{\end{eqnarray}}
\def\be{\begin{equation}}
\def\ee{\end{equation}}
\def\ba{\begin{align}}
\def\ea{\end{align}}
\def\bse{\begin{subequations}}
\def\ese{\end{subequations}}

\newcommand{\nn}{\nonumber}

\def\det{\,{\rm det}\, }

\def\sign{{\rm sgn}}
\def\Span{{\rm Span}}

\def\Ch{{\rm Ch}}

\def\Sym{\,{\rm Sym}\, }

\def\Im{\,{\rm Im}\,}

\DeclareMathOperator{\Erf}{Erf}

\newcommand{\sgn}{\mathop{\rm sgn}}

\def\({\left(}
\def\){\right)}
\def\[{\left[}
\def\]{\right]}
\def\<{\left\langle}
\def\>{\right\rangle}
\def\hf{{1\over 2}}

\newcommand{\p}{\partial}

\newcommand{\eps}{\epsilon}

\newcommand{\de}{\mathrm{d}}

\newcommand{\I}{\mathrm{i}}
\newcommand{\e}{\mathrm{e}}
\newcommand{\rmR}{\mathrm{R}}

\newcommand{\Ssf}{\mathsf{S}}
\newcommand{\esf}{\mathsf{e}}
\newcommand{\ssf}{\mathsf{s}}
\newcommand{\ksf}{\mathsf{k}}

\newcommand{\cD}{\mathcal{D}}

\newcommand{\cV}{\mathcal{V}}
\newcommand{\cC}{\mathcal{C}}
\newcommand{\cS}{\mathcal{S}}

\newcommand{\cB}{\mathcal{B}}
\newcommand{\cK}{\mathcal{K}}
\newcommand{\cM}{\mathcal{M}}

\newcommand{\cE}{\mathcal{E}}

\newcommand{\cT}{\mathcal{T}}
\newcommand{\cJ}{\mathcal{J}}

\newcommand{\cO}{\mathcal{O}}

\newcommand{\IT}{\mathds{T}}
\newcommand{\IR}{\mathds{R}}

\newcommand{\IZ}{\mathds{Z}}

\newcommand{\IN}{\mathds{N}}

\newcommand{\IP}{\mathds{P}}

\def\scR{\mathscr{R}}

\def\Ev{\mathscr{E}}

\def\Zv{\mathscr{Z}}

\def\frK{\mathfrak{K}}

\newcommand{\bfu}{{\boldsymbol u}}
\newcommand{\bfb}{{\boldsymbol b}}

\newcommand{\bfv}{{\boldsymbol v}}

\newcommand{\bfk}{{\boldsymbol k}}

\newcommand{\bfx}{{\boldsymbol x}}
\newcommand{\bfy}{{\boldsymbol y}}

\def\tesf{\tilde\esf}

\def\resf{\mathring{\esf}}

\def\ba{\bar a}

\def\btau{\bar \tau}

\def\hcT{\hat\cT}
\def\hesf{\hat\esf}

\def\hgam{\hat\gamma}

\def\gama{\check\gamma}

\def\CY{\mathfrak{Y}}

\def\ver{\mathfrak{v}}

\def\cEf{\cE^{(0)}}

\def\Ef{\Ev^{(0)}}

\def\rmRi#1{\rmR^{(#1)}}

\def\whh{\widehat h}

\def\EPhi{\Phi^{\,\Ev}}
\def\cEPhi{\Phi^{\,\cE}}

\def\vu{\mathfrak{u}}

\def\cEp{\cE^{(+)}}

\def\Ep{\Ev^{(+)}}

\def\Lv#1{L(#1)}
\def\Rv#1{R(#1)}

\def\sl{a}                    %As slope. Because it appears in k = a n + b
\def\frK{\kappa}

\newcommand{\symfootnote}[1]{%
\let\oldthefootnote=\thefootnote%
\setcounter{mpfootnote}{2}%
\addtocounter{footnote}{-1}%
\renewcommand{\thefootnote}{\fnsymbol{mpfootnote}}%
\footnote{#1}%
\let\thefootnote=\oldthefootnote%
}

\title{Modular anomaly of BPS black holes}

\preprint{arXiv:2204.02207v3}

\author{Sergei Alexandrov and Khalil Bendriss
\\
{\it Laboratoire Charles Coulomb (L2C), Universit\'e de Montpellier,
CNRS, \\ F-34095, Montpellier, France}\\

\vspace*{2mm} {\tt e-mail:
\email{sergey.alexandrov@umontpellier.fr},
\email{khalil.bendriss@umontpellier.fr}
}

\vspace*{-3mm}

}

\abstract{Generating functions of BPS indices, counting states of D4-D2-D0 black holes 
in Calabi-Yau compactifications of type IIA string theory and identified with rank 0 Donaldson-Thomas invariants, 
are examples of mock modular forms. They have a quite complicated modular anomaly expressed as 
a sum over three different types of trees weighted by generalized error functions and their derivatives.
We show that this anomaly can be significantly simplified, which in turn simplfies 
finding the corresponding mock modular generating functions. 
}

\begin{document}

\setlength{\parskip}{0.1cm}

\section{Introduction}

Modularity is an extremely powerful tool to obtain exact results in situations where the modular group, typically $SL(2,\IZ)$, 
is a symmetry of the system. A prominent example is the calculation of black hole entropy in compactified string theory, 
see, e.g., \cite{Strominger:1996sh,Dijkgraaf:1996it,Maldacena:1997de,Dabholkar:2005dt,Gaiotto:2006wm} for some classic works.
In most of these classic examples, the generating function of black hole degeneracies turns out to be a modular form,
sometimes involving a non-trivial multiplier system or of vector valued type.  
But this is not the only possibility as today we know examples where the generating function
has a non-trivial modular anomaly. Usually, the anomaly is of a very specific form so that
the generating function is still considered as a `nice' modular object.

The first such example in the context of black holes was discovered in \cite{Dabholkar:2012nd}
where it was shown that the generating functions of degeneracies of immortal dyons,
i.e. single centered $\frac14$-BPS black holes
appearing in compactifications of type II string theory on $K3\times T^2$, 
the same that have been counted in \cite{Dijkgraaf:1996it},
are {\it mock} modular forms. This is a particularly nice class of functions for which the modular anomaly
is determined by an integral of another modular form called ``shadow" \cite{Zwegers-thesis,MR2605321}.
If the shadow is a product (or sum of products) of a holomorphic and anti-holomorphic modular forms, 
as is the case for \cite{Dabholkar:2012nd}, one says that the mock modular form is mixed.

A few years ago a much richer modular structure was discovered in the context of type IIA string theory
on a Calabi-Yau (CY) threefold $\CY$. 
In \cite{Alexandrov:2018lgp} it was shown that the generating functions of BPS indices 
counting states of D4-D2-D0 black holes and identified with rank 0 Donaldson-Thomas (DT) invariants of $\CY$
are {\it higher depth} mock modular forms, with the depth determined by the D4-brane charge.
This class of functions generalizes mixed mock modular forms and is defined recursively:
a function is mock modular of depth $n$ if its shadow is mock modular of depth $n-1$.
In this terminology the usual modular and mock modular forms have depth 0 and 1, respectively.

The main physical reason for the appearance of mock modularity in both \cite{Dabholkar:2012nd} and 
\cite{Alexandrov:2018lgp}, which describe respectively string theory compactifications with $N=4$ and 
$N=2$ supersymmetry in 4 dimensions, is the existence of multi-centered black holes contributing to the index.
In the $N=4$ context, the only multi-centered black holes are bound states of two $\hf$-BPS black holes, 
while in the $N=2$ setup they can have any number of constituents. This is why the depth of mock modularity 
is 1 in the first case and can take any value in the second.

The work \cite{Alexandrov:2018lgp} not only identified the mathematical nature 
of the generating functions,
but also found the precise anomaly equation which they satisfy
(the cases of depth 0 and 1 have been analyzed before in \cite{Alexandrov:2012au,Alexandrov:2016tnf,Alexandrov:2017qhn}).
This equation takes the form of an expression for {\it modular completion},
a non-holomorphic modification of the corresponding mock modular form
by terms that are suppressed in the limit $\Im\tau\to \infty$
such that it transforms as usual modular form without any anomaly.

This result has already found several applications and has been extended beyond 
the original setup of $N=2$ supergravity.
First, in \cite{Alexandrov:2019rth} it has been generalized to include the so-called refinement,
a one-parameter deformation corresponding physically to switching on the $\Omega$-background \cite{Nekrasov:2002qd}. 
This extension has then been used to find generating functions and their completions of
the (refined) $SU(N)$ Vafa-Witten invariants on $\IP^2$, Hirzebruch
and del Pezzo surfaces \cite{Alexandrov:2020bwg,Alexandrov:2020dyy}, 
which coincide with the (refined) rank 0 DT invariants
on the non-compact CY given by the total space of the canonical bundle over the complex surface
\cite{Minahan:1998vr,Alim:2010cf,Gholampour:2013jfa,gholampour2017localized}. 
The anomaly equation has also been extended to incorporate the cases with larger supersymmetry
such as $N=4$ and $N=8$ \cite{Alexandrov:2020qpb}, 
which allowed to reproduce and generalize the results of \cite{Dabholkar:2012nd} on mock modularity of immortal dyons.
Finally, in \cite{Alexandrov:2023ltz}, using a solution of the modular anomaly found in \cite{Alexandrov:2022pgd},
the generating functions of D4-charge 2 BPS indices for a few one-parameters CY threefolds
have been explicitly found. This generalized the previous works
\cite{Gaiotto:2006wm,Gaiotto:2007cd,Collinucci:2008ht,VanHerck:2009ww,Alexandrov:2023zjb}
where the case of the unit charge, corresponding to depth 0 or standard modular forms, was analyzed
and provided the first examples of mock modular forms assigned to CY threefolds without any additional structures.

A further progress is however complicated by a highly involved form of the modular anomaly or, equivalently, 
of the modular completion of the generating functions.
As we review in section \ref{sec-modcompl}, it is given by a sum over three different types of trees 
used to construct kernels of indefinite theta series from the so-called generalized error functions 
and their derivatives.
The goal of this paper is to improve this situation. 
We show that the expression for the completion found in \cite{Alexandrov:2018lgp}
can actually be significantly simplified. The new expression involves only two types of trees 
and much less terms than the original one. 
(Roughly, the main change is that the functions \eqref{rescEn} are replaced by the functions \eqref{rescEn-new}.)
We hope that this simplification 
will be helpful in solving the anomaly equations for higher depth cases.

The organization of the paper is as follows.
In the next section we present the expression for the completion of generating functions 
obtained in \cite{Alexandrov:2018lgp}. We also find that this expression misses some subtle contributions
which we determine using some remarkable identities satisfied by the generalized error functions.  
In section \ref{ap-exprR} we prove a simplified form of the completion.
In section \ref{sec-concl} we discuss our results. In appendix \ref{ap-generr}
we recall the definition and derive some new properties of the generalized error functions,
including a set of their identities.
In appendix \ref{ap-coef} we collect values of some coefficients entering the completion 
and prove also an identity which they satisfy.
Finally, in appendix \ref{ap-rational} we provide details on how one arrives at the representation
in terms of these coefficients.

\section{Expression for modular completion}
\label{sec-modcompl}

We are interested in D4-D2-D0 BPS states characterized by charge vector $(p^a,q_a,q_0)$
where $a=1,\dots, b_2(\CY)$. The D4-brane charge $p^a$ is an element of the lattice $\Lambda=H_4(\CY,\IZ)$,
while the D2-brane charge $q_a$ belongs to the shifted dual lattice $\Lambda^\star+\hf\, p$.
Given the triple intersection numbers $\kappa_{abc}$ of $\CY$ and D4-brane charge, we define a quadratic form
$\kappa_{ab}=\kappa_{abc}p^c$. It allows to embed $\Lambda$ into $\Lambda^\star$ and represent
\be
q_a=\kappa_{ab}\eps^b+\mu_a +\hf\, \kappa_{ab} p^b,
\label{defmu}
\ee 
where $\eps\in\Lambda$ and $\mu\in \Lambda^\star/\Lambda$ is the residue class which runs over 
$|\det\kappa_{ab}|$ values. The generating functions of D4-D2-D0 BPS indices depend only on 
the D4-charge $p^a$ and the residue class $\mu_a$ and will be denoted $h_{p,\mu}(\tau)$.
According to the results of \cite{Maldacena:1997de,Alexandrov:2012au,Alexandrov:2016tnf,Alexandrov:2018lgp},
under the standard $SL(2,\IZ)$ transformations $\tau\mapsto \frac{a\tau+b}{c\tau+d}$,
they behave as higher depth vector valued mock modular forms
of weight $-\hf b_2-1$ and a multiplier system closely related to 
the Weil representation attached to the lattice $\Lambda$ (see \cite[Eq.(2.10)]{Alexandrov:2019rth}).

A more precise requirement is that the following non-holomorphic function 
\be
\whh_{p,\mu}(\tau,\btau)= h_{p,\mu}(\tau)
+\sum_{n=2}^{n_{\rm max}}
\sum_{\sum_{i=1}^n \gama_i=\gama}
R_n(\{\gama_i\},\tau_2)
\, e^{\pi\I \tau Q_n(\{\gama_i\})}
\prod_{i=1}^n h_{p_i,\mu_i}(\tau)
\label{exp-whh}
\ee
must transform as a modular form of the same weight and multiplier system as $h_{p,\mu}$.
Here the sum runs over all ordered decompositions of $\hgam=(p^a,\mu_a+\frac12 \kappa_{ab} p^b)$
into a sum of $n$ reduced\footnote{We use notation $\gama$ and call it `reduced charge'
because $\gamma$ is usually reserved for the full charge vector that includes also D6 and D0-charges.} 
charge vectors $\hgam_i=(p_i^a,q_{i,a})$
where $q_{i,a}$ are expressed as in \eqref{defmu} with quadratic forms $\kappa_{i,ab}=\kappa_{abc}p_i^c$
and residue classes $\mu_i$, and the first components satisfy (using notation $(xyz)=\kappa_{abc} x^a y^b z^c$)
\be
\label{khcone}
p^3> 0,
\qquad
(r p^2)> 0,
\qquad
k_a p^a > 0,
\ee
for all  divisor classes $r\in H_4^+(\CY,\IZ)$ and curve classes $k \in H_2^+(\CY,\IZ)$.
Physically, this sum corresponds to a sum over all possible bound states of a given charge $\hgam$, 
and the upper limit $n_{\rm max}$ in the sum over the number of constituents $n$
is given by the number of irreducible components which the divisor corresponding to D4-charge $p^a$ 
can be decomposed in.
Besides, in \eqref{exp-whh} we also used notation $Q_n$ for the following quadratic form
\be
Q_n(\{\gama_i\})= \kappa^{ab}q_a q_b-\sum_{i=1}^n\kappa_i^{ab}q_{i,a} q_{i,b} \, ,
\label{defQlr}
\ee
where $\kappa_i^{ab}$ is the inverse of $\kappa_{i,ab}$.
Since $\kappa_{ab}$ is known to have signature $(1,b_2-1)$, the signature of the quadratic form $Q_n$
is $((n-1)(b_2-1),n-1)$. Thus, the sum over D2-brane charges in \eqref{exp-whh}
defines an indefinite theta series of that signature with the kernel given by the function
$R_n$ which is the main non-trivial ingredient of the construction.

\lfig{An example of Schr\"oder tree contributing to $R_8$. 
	Near each vertex we showed the corresponding factor
	using the shorthand notation $\gamma_{i+j}=\gamma_i+\gamma_j$.
}
{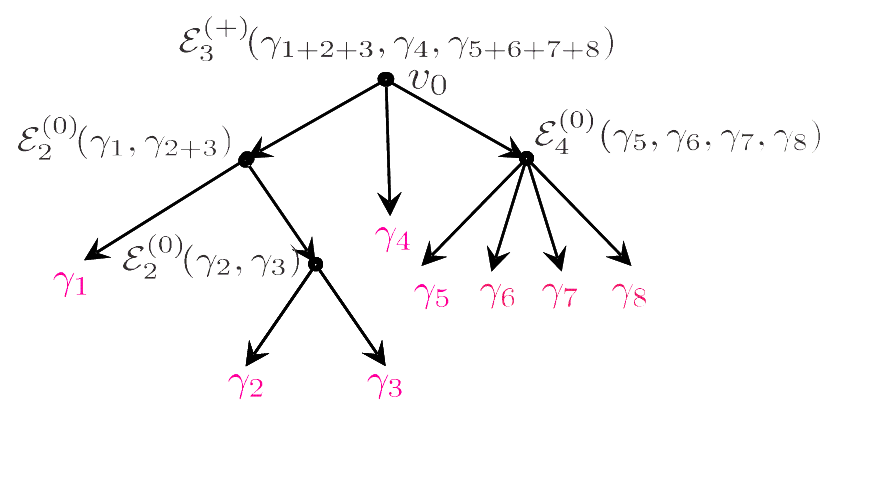}{9.75cm}{fig-Rtree}{-1.2cm}

To define $R_n$, let $\IT_n^{\rm S}$ be the set of the so called 
Schr\"oder trees with $n$ leaves, which are rooted planar trees such that all vertices 
$v\in V_T$ (the set of vertices of $T$ excluding the leaves) 
have $k_v\geq 2$ children (see Fig. \ref{fig-Rtree}). 
We also denote the number of elements in $V_T$ by $n_T$, 
and the root vertex by $v_0$.
The vertices of $T$ are labelled by charges so that the leaves carry charges $\hgam_i$, 
whereas the charges assigned to other vertices
are given recursively by
the sum of charges of their children, $\hgam_v\in\sum_{v'\in\Ch(v)}\hgam_{v'}$.
On top of that, we consider a set of functions $\cE_n(\{\gama_i\};\tau_2)$ 
which have a canonical decomposition
\be
\cE_n(\{\gama_i\};\tau_2)=\cEf_n(\{\gama_i\})+\cEp_n(\{\gama_i\};\tau_2),
\label{twocEs}
\ee
where the first term $\cEf_n$ does not depend on $\tau_2$,
whereas the second term $\cEp_n$ is exponentially suppressed as $\tau_2\equiv\Im\tau\to\infty$ keeping
the charges $\hgam_i$ fixed. The precise definitions of these functions will be given below.
Given a Schr\"oder tree $T$,
we set $\cE_{v}\equiv \cE_{k_v}(\{\hgam_{v'}\})$ (and similarly for $\cEf_{v}, \cEp_{v}$)
where $v'\in \Ch(v)$ runs over the $k_v$ children of the vertex $v$.
In terms of these data, the functions $R_n$ are given by\footnote{Comparing to \cite[Eq.(5.34)]{Alexandrov:2018lgp},
there is an additional factor $2^{1-n}(-1)^{\sum_{i<j} \gamma_{ij} }$.
In \cite{Alexandrov:2018lgp} it was hidden in the functions $\cE_n$, while here we prefer to 
move it to \eqref{solRn}.}  
\be
R_n(\{\gama_i\};\tau_2)=\Sym\Biggl\{\frac{(-1)^{\sum_{i<j} \gamma_{ij} }}{2^{n-1}} \sum_{T\in\IT_n^{\rm S}}(-1)^{n_T-1} 
\cEp_{v_0}\prod_{v\in V_T\setminus{\{v_0\}}}\cEf_{v}\Biggr\},
\label{solRn}
\ee
where $\Sym$ denotes symmetrization (with weight $1/n!$) with respect to the charges $\gama_i$
and 
\be
\gamma_{ij} = p^a_j q_{i,a} - p^a_i q_{j,a} 
\label{def-gammaij}
\ee
is the anti-symmetric Dirac product of charges.
Note that since all $\gamma_{ij}$ are integer and mapped to each other up to sign under permutations,
the first factor is not affected by the symmetrization and can be put in front.

It remains to define the functions $\cE_n$.
To this end, let $\IT_{n,m}^\ell$ be the set of marked unrooted labelled trees with $n$ vertices
and $m$ marks where the marks are assigned to vertices and each mark adds 2 labels to its vertex.
For example, $|\IT_{2,1}^\ell|=4$ with all trees having the same topology (2 vertices, one of which carries a mark),
4 labels (3 at the marked vertex and 1 at the non-marked vertex) and differing by the label assigned 
to the non-marked vertex.
In our case the labels are identified with the charges. Thus, given a set $\{\gama_1,\dots,\gama_{n+2m}\}$, 
the trees are decorated in the following way. 
Let $m_\ver\in \{0,\dots m\}$ be the number of marks carried by the vertex $\ver$, 
so that $\sum_\ver m_\ver=m$. Then a vertex $\ver$ with $m_\ver$ marks carries 
$1+2m_\ver$ charges $\gama_{\ver,s}$, $s=1,\dots,1+2m_\ver$ and 
we set $\gama_\ver=\sum_{s=1}^{1+2m_\ver}\gama_{\ver,s}$.
Given a tree $\cT\in \IT_{n,m}^\ell$, we denote the set of its edges by $E_{\cT}$, the set of vertices by $V_{\cT}$, 
the source and target vertex\footnote{The orientation 
	of edges on a given tree can be chosen arbitrarily since the final result does not depend on this choice,
	although various intermediate quantities can flip sign under the flip of the orientation.} 
of an edge $e$ by $s(e)$ and $t(e)$, respectively,
and the two disconnected trees obtained from $\cT$ by removing the edge $e$ by $\cT_e^s$ and $\cT_e^t$. 
Furthermore, we will use the boldface script to denote $nb_2$-dimensional vectors and 
assign to each edge the vector 
\be
\bfv_e=\sum_{i\in V_{\cT_e^s}}\sum_{j\in V_{\cT_e^t}}\bfv_{ij},
\label{defue}
\ee
where $\bfv_{ij}$ are the vectors with the following components
\be
v_{ij,k}^a=\delta_{ki} p^a_j-\delta_{kj} p^a_i.
\label{defvij}
\ee

Using these notations, we define 
\cite[Eq.(5.32)]{Alexandrov:2018lgp}
\be
	\cEPhi_n(\bfx)=
	\frac{1}{ n!}\sum_{m=0}^{[(n-1)/2]}
	\sum_{\cT\in\, \IT_{n-2m,m}^\ell}
	\[\prod_{\ver\in V_\cT}\cD_{m_\ver}(\{p_{\ver,s}\})\prod_{e\in E_\cT} \cD(\bfv_{s(e) t(e)})\]
	\Phi^E_{\cT}(\bfx).
\label{rescEn}
\ee
Here $\Phi^E_{\cT}$ is a generalized error function defined by an unrooted labelled tree 
in terms of the usual (boosted) generalized error function $\Phi^E_n$ with parameters given by the vectors \eqref{defue},
\be 
\Phi^E_{\cT}(\bfx)=\Phi_{|E_\cT|}^E(\{\bfv_e\};\bfx).
\label{defPhiT-main}
\ee 
The functions $\Phi^E_n$ are reviewed
in appendix \ref{ap-generr}  and defined with respect to the bilinear form
\be
\bfx\cdot\bfy= \sum_{i=1}^n \kappa_{abc}p_i^a x_i^b y_i^c.
\label{biform}
\ee
Besides, $\cD(\bfv)$ is the differential operator \eqref{defcD}
and $\cD_{m}(\{p_{s}\})$ is another differential operator associated with the existence of marks.
It is given by a sum over unrooted labelled trees with $2m+1$ vertices. To write the precise formula, 
we need to introduce rational coefficients $a_\cT$ determined recursively by the relation
\be
a_\cT=\frac{1}{n_\cT}\sum_{\ver\in V_\cT} (-1)^{n_\ver^+} \prod_{s=1}^{n_\ver} a_{\cT_s(\ver)},
\label{res-aT}
\ee
where $n_\cT$ is the number of vertices, $n_\ver$ is the valency of the vertex $\ver$,
$n_\ver^+$ is the number of incoming edges at the vertex, and
$\cT_s(\ver)$ are the trees obtained from $\cT$ by removing the vertex.
The relation \eqref{res-aT} is supplemented by the initial condition for the trivial tree 
consisting of one vertex, $a_\bullet=1$, and one can show that $a_\cT=0$ for all trees with even number of vertices.
(We provide a table of these coefficients for trees with $n_\cT\leq 7$ in appendix \ref{ap-coef}.)
Then we have
\be
\cD_{m}(\{p_{s}\})=\sum_{\cT\in\, \IT_{2m+1}^\ell} a_{\cT}\prod_{e\in E_{\cT}}\cD(\bfv_{s(e)t(v)}).
\label{defcDcT}
\ee
Note that all vectors entering the definition of $\cD_{m_\ver}$ are orthogonal to the vectors 
appearing as parameters of the operators $\cD$ and the generalized error functions in \eqref{rescEn}.
Therefore, $\cD_{m_\ver}$ do not actually act on those factors and can be replaced by
usual functions obtained as $\cD_{m_\ver}\cdot 1$.

In terms of the functions \eqref{rescEn}, we finally set
\be
\cE_n(\{\gama_i\};\tau_2)=
\frac{\cEPhi_n(\bfx)}{(\sqrt{2\tau_2})^{n-1}}\, ,
\label{rescEnPhi}
\ee
where $\bfx$ is a vector with components $x^a_i=\sqrt{2\tau_2}\, \kappa_i^{ab}q_{i,b}$.
In particular, its scalar products with the vectors $\bfv_{ij}$ \eqref{defvij}
reproduce the Dirac products \eqref{def-gammaij}, i.e. $\bfv_{ij}\cdot\bfx=\sqrt{2\tau_2}\gamma_{ij}$.

This completes the definition of the functions $R_n$ determining the modular anomaly 
of the generating functions via \eqref{exp-whh}. Given that their building blocks 
$\cE_n$ are constructed from the generalized error functions with parameters defined by charges $\gama_i$,
they have the meaning of kernels of indefinite theta series providing completions 
for holomorphic theta series constructed from sign functions of Dirac products of these charges.

\subsection{Large $\tau_2$ limit and vanishing Dirac products}
\label{subsec-degen}

Although the above construction provides a complete definition of the functions $R_n$,
to actually compute them, one should perform the split \eqref{twocEs}, which amounts 
to evaluating the large $\tau_2$ limit of the functions $\cE_n$ \eqref{rescEnPhi}.
This has been done in \cite[Eq.(5.29)]{Alexandrov:2018lgp} with the following result
\be
\begin{split}
\cEf_n(\{\gama_i\})\equiv&\,  \lim_{\tau_2\to\infty}\cE_n(\{\gama_i\};\tau_2)
\\
=&\, \frac{1}{n!}
\sum_{m=0}^{[(n-1)/2]}\sum_{\cT\in\, \IT_{n-2m,m}^\ell}\frK_\cT(\{\gama_{\ver}\})
\prod_{\ver\in V_\cT}\cV_{m_\ver}(\{\gama_{\ver,s}\})
\prod_{e\in E_{\cT}}\sign(\Gamma_e),
\end{split}
\label{defEn0new}
\ee
where
\be
\Gamma_e=\sum_{i\in V_{\cT_e^s}}\sum_{j\in V_{\cT_e^t}}\gamma_{ij},
\qquad
\frK_\cT(\{\gama_{\ver}\})=\prod_{e\in E_{\cT}}\gamma_{s(e) t(e)},
\label{defGame}
\ee
and 
\be
\cV_{m}(\{\gama_{s}\})=\sum_{\cT\in\, \IT_{2m+1}^\ell} a_{\cT}\frK_\cT(\{\gama_{s}\}).
\label{deftcV-mw}
\ee

However, this result is not quite precise. It was derived from the fact that the generalized error functions
$\Phi_n^E(\{\bfv_i\};\bfx)$ at large values of their argument $\bfx$ 
reduce to a product of sign functions (see appendix \ref{ap-generr} below \eqref{generrPhiME}). 
But along certain directions, characterized by vanishing of two or more scalar products 
$\bfv_i\cdot \bfx$, this is not true, at least if we accept the convention that 
$\sgn(0)=0$.\footnote{This issue has already been realized in \cite{Alexandrov:2019rth} where an improved solution 
has been proposed in the refined case, but it was claimed that the unrefined case does not suffer from this problem.
Our analysis shows that this claim was naive.}
The correct limit is given by Proposition \ref{prop-largex} in appendix \ref{ap-generr}.
From this result it is easy to derive a corrected version of \eqref{defEn0new}. 
Essentially, the derivation given in the proof of \cite[Prop.5]{Alexandrov:2018lgp} goes through upon 
replacing $\prod_{e\in E_{\cT}}\sign(\Gamma_e)$ by 
\be 
\Ssf_\cT(\{\gama_\ver\})=
\sum_{\cJ\subseteq E_\cT}\esf_{\cT_\cJ}(\{\bfv_e\}_{e\in\cJ})
\,\prod_{e\in \cJ}\delta_{\Gamma_e}
\prod_{e\in E_\cT\setminus \cJ} \sgn (\Gamma_e),
\label{exprSsf}
\ee
where\footnote{We left the dependence of the coefficients $\esf_\cT$ on the vectors $\bfv_e$ explicit
to distinguish them from other coefficients $e_\cT$ appearing below, which depend only on the tree topology,
but not on its labelling by charges.} 
\be 
\esf_\cT(\{\bfv_e\})=\Phi_{\cT}^E(0)
\ee
and $\cT_\cJ$ is the tree obtained from $\cT$ by contracting the edges $e\in E_\cT\backslash \cJ$.
The crucial observation for this is that the coefficients $\esf_\cT$ satisfy
\be
\label{esfT-T123id}
\esf_{\hat{\cT}_1}+\esf_{\hat{\cT}_2}-\esf_{\hat{\cT}_3}=-\esf_{\cT},
\ee
where $\hcT_r$ ($r=1,2,3$) are the trees constructed from arbitrary unrooted trees $\cT_r$ as 
shown in Fig. \ref{fig-Vign3}, while $\cT$ is obtained from any tree $\hcT_r$ by collapsing both edges $e_i$. 
This relation is direct consequence of the identity \eqref{PhiT-T123id} for the generalized error functions
specified to $\bfx=0$.
As a result, one finds the following expression for the large $\tau_2$ limit of the functions $\cE_n$:
\be
\cEf_n(\{\gama_i\})
= \frac{1}{n!}
	\sum_{m=0}^{[(n-1)/2]}\sum_{\cT\in\, \IT_{n-2m,m}^\ell}
	\[\prod_{\ver\in V_\cT}\cV_{m_\ver}(\{\gama_{\ver,s}\})\]\frK_\cT(\{\gama_{\ver}\})\,\Ssf_\cT(\{\gama_\ver\}).
\label{En0new}
\ee

\lfig{Unrooted trees constructed from the same three subtrees.}
{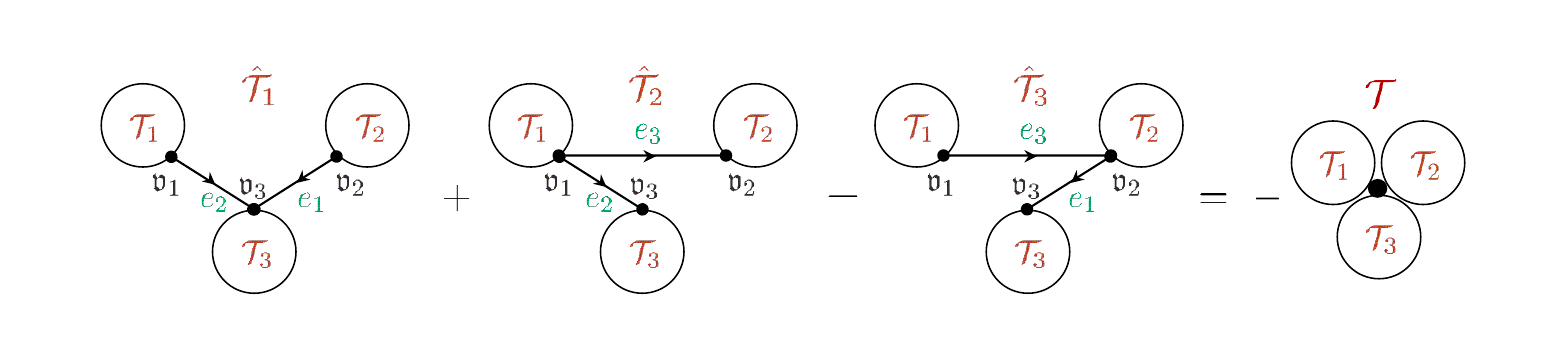}{19.5cm}{fig-Vign3}{-1.cm}

It is useful to note that the sum over marked trees can also be rewritten in terms of a sum over subsets of 
$\Zv_{n}=\{1,\dots,n\}$ and a single function given by the sum over non-marked trees. 
Namely, let us introduce 
\be 
\Ef_n(\{\gama_i\})=\frac{1}{n!}\sum_{\cT\in\, \IT_n^\ell}
\frK_\cT(\{\gama_{\ver}\})\, \Ssf_\cT(\{\gama_\ver\}),
\label{exprcEf-new}
\ee
where the sum goes over labelled trees without marks,
which is nothing but the $m=0$ part of $\cEf_n$.
Let also $\{\cJ_j\}$ be a set of non-intersecting subsets of $\Zv_{n}$ 
of odd cardinality $|\cJ_j|\equiv 2m_j+1$, $m_j\in \IN$, and $\gama^\cJ_j=\sum_{i\in\cJ_j}\gama_i$.
Then it is easy to see that \eqref{En0new} is equivalent to 
\be
\cEf_n(\{\gama_i\};\tau_2)=\sum_{m=0}^{[(n-1)/2]}\sum_{\cup_{j=1}^{n-2m}\cJ_j=\Zv_{n}}
\[\prod_{j=1}^{n-2m}\cV_{m_j}(\{\gama_i\}_{i\in\cJ_j})\]\Ef_{n-2m}(\{\gama^\cJ_j\};\tau_2).
\label{cEf-subset}
\ee
A similar representation also exists for the full function $\cE_n$ defined by \eqref{rescEn}
and \eqref{rescEnPhi}, but we will not use it in this paper.

In principle, \eqref{En0new} or \eqref{cEf-subset} solves the problem of correcting \eqref{defEn0new}.
Note, however, that while for even number of vertices (and hence odd number of edges) of $\cT$ 
the coefficients $\esf_\cT$ vanish due to the oddness of the corresponding generalized error functions (see \eqref{oddPhi}),
in the opposite case they are, in general, irrational numbers.\footnote{For generic tree, $\esf_\cT$
	is irrational even being computed on the simplest set of charges all taken to be equal. A notable exception
	is the tree of linear topology $\bullet\!\mbox{---}\!\bullet\!\mbox{--}\cdots \mbox{--}\!\bullet\!\mbox{---}\!\bullet\,$ 
	for which $\esf_\cT$ is equal to $1/n_\cT$ if $n_\cT$ is odd and 0 otherwise, 
	as has been observed in \cite{Alexandrov:2019rth}.
	In fact, this is a consequence of the identity \eqref{ident-eT} which holds for arbitrary charges.}
On the other hand, only rational numbers 
are expected to appear in the completion of a mock modular form with rational Fourier coefficients.
This implies that there should exist an alternative representation 
and, on the basis of this expectation, we suggest the following 

\begin{conj}
\label{conj-eT}
\be
\Ef_n(\{\gama_i\})
= \frac{1}{n!}
\sum_{\cT\in\, \IT_n^\ell} \frK_\cT(\{\gama_{\ver}\})\, S_\cT(\{\gama_\ver\}),
\label{En0new2}
\ee	
where
\be 
S_\cT(\{\gama_\ver\})=
\sum_{\cJ\subseteq E_\cT}e_{\cT_\cJ}
\,\prod_{e\in \cJ}\delta_{\Gamma_e}
\prod_{e\in E_\cT\setminus \cJ} \sgn (\Gamma_e)
\label{exprS}
\ee 
and $e_{\cT}$ are some rational numbers depending only on topology of $\cT$.
\end{conj}

It is clear that if this conjecture is true then a similar replacement of $\Ssf_\cT$ by $S_\cT$ 
can also be made in \eqref{En0new}.
In appendix \ref{ap-rational} we provide an evidence in favor of this conjecture by explicitly 
analyzing the cases with $n\leq 5$. In all these cases the equality of \eqref{exprcEf-new} and \eqref{En0new2} 
follows from a set of remarkable identities satisfied by the generalized error functions (see \S\ref{subsec-identPhi})
and hence by the coefficients $\esf_\cT$.
The simplest of them holds for the tree of trivial topology
$\cT=\bullet\!\mbox{---}\!\bullet\!\mbox{--}\cdots \mbox{--}\!\bullet\!\mbox{---}\!\bullet\,$
with odd $n_\cT$
and takes the following form
\be 
\sum_{\rm cyclic\ perm.}\esf_\cT(\{\bfv_e\})=1,
\label{ident-eT}
\ee
where the sum goes over the cyclic permutations of charges assigned to the vertices of $\cT$.
This identity implies that for this tree $e_\cT=1/n_\cT$, consistently with the solution found in \cite{Alexandrov:2019rth}
in the refined case where the tree of linear topology is the only one contributing to the refined
analogue of $\cE_n$. 

For practical purposes, Conjecture \ref{conj-eT} is useless until one finds concrete values of the coefficients $e_\cT$
for all unrooted trees. In fact, these coefficients are not even unambiguously defined because, as shown in 
appendix \ref{ap-rational}, different terms in \eqref{En0new2} are not independent. 
Nevertheless, as we will now explain, there is a way to determine (a possible set of) their values.
To this end, we first prove 

\begin{proposition}
\label{prop-vanishG}
If $\Gamma_e=0$ for $\forall e\in E_\cT $, then this is true for any unrooted tree with the same number of vertices
labelled by the same set of charges,
i.e. $\Gamma_e=0$ for $\forall e\in E_{\cT'} $ and $\forall \cT'$ such that
$n_{\cT'}=n_\cT$.
\end{proposition}
\begin{proof}
Let us consider unrooted trees with $n$ vertices and think about $\gamma_{ij}$ as basis vectors in 
a $\hf n(n-1)$-dimensional vector space. Then for each tree $\cT\in\, \IT_{n}^\ell$, the set of $\Gamma_e$ 
spans a $(n-1)$-dimensional subspace. What we need to show is that this subspace is the same for all trees.
In turn, this would follow if for any pair of trees, $\cT$ and $\cT'$, all $\Gamma'_e$ can be expressed 
as linear combinations of $\Gamma_e$. Moreover, due to the linear independence of these vectors for each tree,
it is enough to show this for a fixed tree $\cT'$.

\lfig{A vertex and adjacent subtrees.}
{tree-vi}{6cm}{fig-tree-vi}{-0.7cm}

Let us choose $\cT'$ to be the tree where vertex $\ver_n$ is connected to all other vertices.
(For example, in the table of appendix \ref{ap-coef}, these are the trees $\cT_{3,1}$, $\cT_{5,3}$ and $\cT_{7,8}$.) 
For such tree the edges can be labelled by the index of the vertex connected by the edge to $\ver_n$
and we have $\Gamma'_{e_i}=\sum_{j=1}^n \gamma_{ij}$, $i=1,\dots,n-1$.
On the other hand, for any other unrooted labelled tree $\cT$, let 
\be
M_{ie}=\left\{\begin{array}{ll}
	1 &\quad \mbox{if } \ver_i=s(e),
	\\ 
	-1 &\quad \mbox{if } \ver_i=t(e),
	\\
	0 &\quad \mbox{if } \ver_i\notin e.
\end{array} \right.
\ee
Then it is easy to realize that 
\be 
\Gamma'_{e_i}=\sum_{e\in E_\cT} M_{ie}\, \Gamma_e.
\label{eqGG}
\ee
Indeed, let us denote the edges adjacent to $\ver_i$ by $e_r$, $r=1,\dots, m=n_{\ver_i}$, and the subtree
joint to $\ver_i$ via $e_r$ by $\cT_r$ (see Fig. \ref{fig-tree-vi}).
With these notations, the definitions of $\Gamma_e$ \eqref{defGame} and the matrix $M_{ie}$ imply that
\be 
\sum_{e\in E_\cT} M_{ie}\, \Gamma_e=\sum_{r=1}^m \sum_{j\in V_{\cT_r}}
\Biggl(\gamma_{ij}+\sum_{s=1\atop s\ne r}^m\sum_{k\in V_{\cT_s}}\gamma_{kj}\Biggr).
\ee
The second term on the r.h.s. vanishes because $\gamma_{kj}$ is anti-symmetric, while the first equals $\Gamma'_{e_i}$. 
This proves \eqref{eqGG}, from which the Proposition follows according to the above reasoning.
\end{proof}

Due to this proposition, the condition of vanishing of all $\Gamma_e$ is the same for all trees
with a given number of vertices. Let us denote by $\gama_i^\star$ the charges satisfying this condition.
We claim that 
\be 
\cEf_n(\{\gama_i^\star\})=0.
\label{vanish-cE0}
\ee
This fact is a consequence of a (conjectural) formula (5.47) from \cite{Alexandrov:2018lgp}
for $\cEf_n$. In this representation each product of sign functions is multiplied by a factor
\be
\sum_{T}d_T\,\kappa(T),
\qquad
\kappa(T) =(-1)^{n-1} \prod_{v\in V_T}  (-1)^{\gamma_{\Lv{v}\Rv{v}}} \gamma_{\Lv{v}\Rv{v}},
\label{PPn}
\ee
where $d_T$ are some numerical coefficients and the sum goes over a subset of flow trees $T$ with $n$ leaves,
which are rooted binary trees with vertices decorated by charges $\gamma_v$,
such that the leaves of the tree carry charges $\gamma_i$, and
the charges propagate along the tree according to $\gamma_v= \gamma_{\Lv{v}}+\gamma_{\Rv{v}}$
at each vertex, where $\Lv{v}$, $\Rv{v}$ are the two children of the vertex $v$.
It is clear that the factor corresponding to the root vertex of a tree $T$ in the product $\kappa(T)$
coincides with one of $\Gamma_e$'s and hence all terms vanish for our choice of charges.

Applying the vanishing condition \eqref{vanish-cE0} to \eqref{En0new} with $\Ssf_\cT$ replaced by $S_\cT$ gives
\be
\sum_{m=0}^{[(n-1)/2]}\sum_{\cT\in\, \IT_{n-2m,m}^\ell}e_\cT\frK_\cT(\{\gama_{\ver}^\star\})
\prod_{\ver\in V_\cT}\cV_{m_\ver}(\{\gama_{\ver,s}^\star\})=0.
\label{En0new2-vanish}
\ee	
Substituting \eqref{deftcV-mw}, this can be rearranged as 
\be
\sum_{\cT\in\, \IT_{n}^\ell}
\(\sum_{\smash{\mathop{\cup}\limits_{k=1}^m\cT_k \simeq\cT }}\,
e_{\cT/\{\cT_k\}}
\prod_{k=1}^m a_{\cT_k}\)
\frK_\cT(\{\gama_{\ver}^\star\})=0,
\label{En0new2-vanish2}
\ee	
where the second sum runs over all decompositions of $\cT$ into a set of 
non-intersecting subtrees\footnote{We use the sign $\simeq$ because a tree contains more information 
than the union of its subtrees.}
and $\cT/\{\cT_k\}$ denotes the tree obtained from $\cT$ by collapsing each subtree $\cT_k$ to a single vertex.
Although the factors $\frK_\cT(\{\gama_{\ver}^\star\})$ are not all linearly independent,
the condition \eqref{En0new2-vanish2} suggests a simple way to fix all the coefficients $e_\cT$ 
by equating to zero each term in the sum. Then, taking into account that for $m=n_\cT$ one has $\cT_k=\bullet$ and
$\cT/\{\cT_k\}=\cT$, one obtains the following recursive formula
\be 
e_\cT=-\sum_{m=1}^{n_\cT-1}\sum_{\smash{\mathop{\cup}\limits_{k=1}^m\cT_k \simeq\cT }}\,
e_{\cT/\{\cT_k\}}\prod_{k=1}^m a_{\cT_k}.
\label{res-eT}
\ee 
It is easy to see that. similarly to $a_\cT$, the coefficients vanish for even number of vertices.
For trees with $n_\cT\leq 7$ we collected them in a table in appendix \ref{ap-coef}
and the resulting values are consistent with the analysis in appendix \ref{ap-rational}.
It is worth noting also that the coefficients $e_\cT$ satisfy a relation analogous to \eqref{esfT-T123id}:
\be
\label{eT-T123id}
e_{\hat{\cT}_1}+e_{\hat{\cT}_2}-e_{\hat{\cT}_3}=-e_{\cT},
\ee
where the trees are as in Fig. \ref{fig-Vign3}.
We provide its proof in appendix \ref{ap-coef}.

\section{Simplified anomaly}
\label{ap-exprR}

In this section we show that the expression for the modular completion reviewed in the previous section
can be significantly simplified. While the equations \eqref{exp-whh} and \eqref{solRn} essentially remain intact,
the simplification happens at the level of their building blocks given by the functions $\cE_n$, which get replaced by 
simpler objects which we call $\Ev_n$. 
More precisely, we claim that the functions $R_n$ \eqref{solRn} can be rewritten as 
\be
R_n(\{\gama_i\};\tau_2)=\Sym\Biggl\{\frac{(-1)^{\sum_{i<j} \gamma_{ij} }}{2^{n-1}} \sum_{T\in\IT_n^{\rm S}}(-1)^{n_T-1} 
\Ep_{v_0}\prod_{v\in V_T\setminus{\{v_0\}}}\Ef_{v}\Biggr\},
\label{solRn-new}
\ee
where $\Ef_n$ and $\Ep_n$ are the constant and exponentially suppressed terms in a decomposition of 
the new functions $\Ev_n$,
\be
\Ev_n(\{\gama_i\};\tau_2)=\Ef_n(\{\gama_i\})+\Ep_n(\{\gama_i\};\tau_2),
\label{twocEs-new}
\ee
These functions are defined, similarly to \eqref{rescEnPhi}, as 
\be
\Ev_n(\{\gama_i\};\tau_2)=
\frac{\EPhi_n(\bfx)}{(\sqrt{2\tau_2})^{n-1}}\, ,
\label{rescEnPhi-new}
\ee
where $\EPhi_n(\bfx)$ replace $\cEPhi_n(\bfx)$ and take the following form
\be
\EPhi_n(\bfx)=
\frac{1}{n!}\sum_{\cT\in\, \IT_n^\ell}
\[\prod_{e\in E_\cT} \cD(\bfv_{s(e) t(e)},\bfy)\]
\Phi^E_{\cT}(\bfx)\Bigr|_{\bfy=\bfx},
\label{rescEn-new}
\ee
written in terms of a differential operator generalizing the one in \eqref{defcD} and given by
\be
\cD(\bfv,\bfy)=\bfv\cdot\(\bfy+\frac{1}{2\pi}\,\p_\bfx\).
\label{defcDif}
\ee
Furthermore, the constant part of $\Ev_n$ coincides with the function $\Ef_n$ \eqref{En0new2}.

Comparing \eqref{rescEn-new} and \eqref{En0new2} to \eqref{rescEn} and \eqref{cEf-subset}, 
we observe that the simplification consists in dropping all contributions generated by trees with 
non-vanishing number of marks and by the mutual action of the derivative operators $\cD(\bfv)$.
The latter is achieved by the use of the operators \eqref{defcDif} which do not act on each other in 
contrast to the original operators $\cD(\bfv)$. The price to pay for this is that the functions 
$\EPhi_n(\bfx)$ are {\it not} eigenfunctions of the Vign\'eras operator \eqref{Vigdif} and 
therefore appear to spoil modularity. However, the claim is that  
all anomalies are cancelled in the combinations relevant for the completion.

The proof of the representation \eqref{solRn-new} proceeds in two steps.
At the first step we note that the functions $\cE_n$ \eqref{rescEnPhi} can be expressed as follows
\be
\begin{split}
\cE_n(\{\gama_i\};\tau_2)=&\, 	\frac{1}{n!}\sum_{m=0}^{[(n-1)/2]}
\frac{1}{(\sqrt{2\tau_2})^{n-2m-1}} 
\\
&\times
\sum_{\cT\in\, \IT_{n-2m,m}^\ell}
\[ \prod_{\ver \in V_\cT} \cV_{m_\ver}(\{\hgam_{\ver,s}\})\prod_{e\in E_\cT}\cD(\bfv_{s(e) t(e)},\bfy) \]
\Phi_{\cT}^E(\bfx)\Bigr|_{\bfy=\bfx}\, ,
\end{split}
\label{gen-oldident}
\ee
where, as usual, $\bfx$ is a vector with components $x^a_i=\sqrt{2\tau_2}\, \kappa_i^{ab}q_{i,b}$.
This relation generalizes the one in \eqref{En0new} to finite $\tau_2$.
Comparing to the original representation using the functions $\cEPhi_n(\bfx)$ \eqref{rescEn}, there are two changes:
the functions $\cD_{m}$ (see the comment below \eqref{defcDcT}) are replaced by $(2\tau_2)^m\cV_{m}$, and 
the operators $\cD(\bfv)$ are replaced by $\cD(\bfv,\bfy)$.
The proof of this statement is completely analogous to the proof of Proposition 5 in \cite{Alexandrov:2018lgp}
where instead of \cite[Eq.(F.20)]{Alexandrov:2018lgp} one should use the identity \eqref{PhiT-T123id} and 
make the replacement
\be
\prod_{e\in E_{\cT}}(\bfv_{s(e) t(e)},\bfx)\,\sgn(\bfv_e,\bfx)
\ \longrightarrow\ 
\[\prod_{e\in E_\cT}\cD(\bfv_{s(e) t(e)},\bfy) \]
\Phi_{\cT}^E(\bfx)\Bigr|_{\bfy=\bfx}.
\ee

\lfig{Combination of two Schr\"oder trees ensuring the cancellation of contributions generated by marked trees.}
{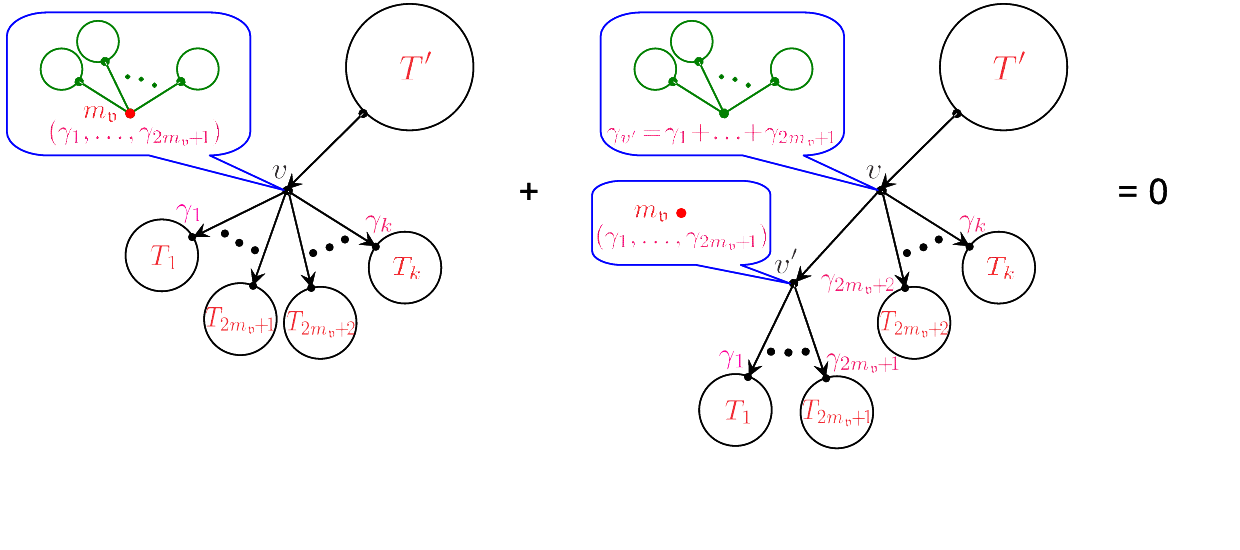}{17.cm}{fig-Mcancel}{-1.5cm}

The second step is essentially identical to the proof of Proposition 10 in \cite{Alexandrov:2018lgp}.
Namely, let us consider the original representation of $R_n$ \eqref{solRn} 
and pick up the contribution generated by a non-trivial marked tree $\cT$, 
i.e. a tree having more than one vertex and at least one mark, 
which appears in the sum over marked unrooted trees producing the factor assigned
to a vertex $v$ of a Schr\"oder tree $T$. 
We denote $k=n_v$ the number of children of the vertex $v\in V_T$ and $\gamma_i$ ($i=1,\dots,k$) their charges.
Let us focus on a vertex $\ver\in V_\cT$ with $m_\ver>0$ marks and take 
$\gamma_s$ ($s=1,\dots,2m_\ver+1$) to be the charges labelling this vertex, 
so that its weight to our contribution, due to \eqref{En0new} and \eqref{gen-oldident}, 
is given by $\cV_{m_\ver}(\{\hgam_{s}\})$.
Note that $k\ge 2m_\ver+2$ because the tree $\cT$ 
has at least one additional vertex except $\ver$.
Then it is easy to see that the contribution we described is exactly cancelled 
by the contribution coming from another Schr\"oder tree,
which is obtained from $T$ by adding an edge
connecting the vertex $v$ to a new vertex $v'$, whose children are the $2m_\ver+1$ children of $v$ in $T$ 
carrying charges $\gamma_s$ (see Fig. \ref{fig-Mcancel}).\footnote{The new tree is of Schr\"oder 
	type because its vertex $v$ has $k-2m_{\ver}\ge 2$ children and vertex $v'$ has $2m_\ver+1\ge 3$ children.}
Indeed, choosing in the sum over marked trees at vertex $v$ a tree $\cT$ which is the same as before except
that now it has 0 marks at vertex $\ver$, and
in the sum at vertex $v'$ the trivial tree having one vertex and $m_\ver$ marks, 
one gets exactly the same contribution as before,
but now with an opposite sign due to the presence of an additional vertex in the Schr\"oder tree.
Thus, all contributions from non-trivial marked trees are cancelled.

As a result, we remain only with the contributions generated by trivial marked trees, 
i.e. having only one vertex and $m_\ver$ marks.
One has to distinguish two cases: either the corresponding vertex $v$ of the Schr\"oder tree is the root or not.
In the former case, this contribution is trivially cancelled in the difference $\cE_{v_0}-\cEf_{v_0}=\cEp_{v_0}$ 
assigned to $v_0$ in \eqref{solRn}.
In the latter case, there are again two possibilities: whether the vertex $\ver$ of the unrooted tree $\cT$ assigned to
the parent of $v$, which carries the charge $\gamma_v$, has marks or not.
If not, this is precisely the contribution used above to cancel the contributions
from non-trivial marked trees. If $m_\ver>0$, then $\cT$ must be a trivial marked tree because 
the contributions corresponding to non-trivial ones have already been cancelled. 
But then we can apply the same argument to this parent vertex. 
As a result, we continue in this way up to the root, but  
we already know that its weight $\cEp_{v_0}$ does not contain the contributions of trivial marked trees.
Thus, all contributions generated by unrooted trees in \eqref{En0new} and \eqref{gen-oldident}  
with at least one mark cancel and we arrive at the formula \eqref{solRn-new}.

\section{Conclusions}
\label{sec-concl}

The main result of this paper is a simplified expression for the modular completion of generating series 
$h_{p,\mu}(\tau)$ of D4-D2-D0 BPS indices
which are identical to rank 0 DT invariants of a CY threefold.
This completion, which determines the modular anomaly of $h_{p,\mu}(\tau)$, is characterized by functions 
$R_n(\{\gama_i\},\tau_2)$ given by a sum over Schr\"oder trees $\IT_n^{\rm S}$. 
Before simplifications the weights $\cE_n$ 
assigned to vertices of these trees were given themselves by a sum over marked unrooted labelled trees $\IT_{n,m}^\ell$,
with the weights assigned to the marks also given by a sum over unrooted labelled trees $\IT_{n}^\ell$.
After the simplification we remained only with a single sum over $\IT_{n}^\ell$, 
i.e. all contributions of marks have been cancelled, which significantly simplifies the final expression
and makes it much easier to perform formal manipulations with the anomaly.   

The equations determining the simplified completion are \eqref{exp-whh}, \eqref{solRn-new}, \eqref{rescEnPhi-new} 
and \eqref{rescEn-new}. Sometimes it is convenient to split the sum over reduced charges $\gama_i$ 
appearing in \eqref{exp-whh} into sums over their magnetic and electric components, and to move the symmetrization
over charges and the sign factor, appearing in \eqref{solRn-new} together with the sum over Schr\"oder trees, 
to the sum over electric charges.
Therefore, we give here an alternative version of the expression for the modular completion:
\be
\whh_{p,\mu}(\tau)=h_{r,\mu}(\tau)+ \sum_{n=2}^{n_{\rm max}}\sum_{\sum_{i=1}^n p_i=p}
\sum_{\{\mu_i\}}
\rmRi{\{p_i\}}_{\mu,\{\mu_i\}}(\tau, \btau)
\prod_{i=1}^n h_{p_i,\mu_i}(\tau),
\label{exp-whh-new}
\ee
where 
\be
\rmRi{\{p_i\}}_{\mu,\{\mu_i\}}(\tau, \btau)=
\sum_{\sum_{i=1}^n q_i=\mu+\hf p } 
\Sym\Bigl\{(-1)^{\sum_{i<j} \gamma_{ij} }\scR_n(\{\gama_i\};\tau_2)\, e^{\pi\I \tau Q_n(\{\gama_i\})}\Bigr\}
\label{defRn}
\ee
and
\be
\scR_n(\{\gama_i\};\tau_2)=\frac{1}{2^{n-1}} \sum_{T\in\IT_n^{\rm S}}(-1)^{n_T-1} 
\Ep_{v_0}\prod_{v\in V_T\setminus{\{v_0\}}}\Ef_{v} .
\label{solRn-new2}
\ee

Besides, we have found that the previous results on the large $\tau_2$ limit of the functions 
determining the weights assigned to vertices of Schr\"oder trees in \eqref{solRn-new2} 
must be corrected as they were missing
contributions from certain degenerate charge configurations. The correct limit can be found in \eqref{En0new2}
and is given in terms of rational coefficients $e_\cT$ assigned to unrooted trees.
Although we formulated it as a conjecture, we have verified it analytically for $n\leq 5$ 
and also did some numerical tests for $n=7$. For this purpose, we have derived new interesting identities between
the generalized error functions, which are the main building blocks of our construction (see \S\ref{subsec-identPhi}). 
It appears that all these non-trivial identities
are consequences of a single identity \eqref{identPhiE}.

These results are of direct interest for any attempt to find explicit expressions 
of the generating functions $h_{p,\mu}$, especially when magnetic charges are not small, so that
the depth of mock modularity is not 0 or 1, and one cannot rely on generating functions of {\it refined}
invariants, as is the case of {\it compact} CY threefolds
where the precise mathematical definition of refined DT invariants remains obscure.
In such cases we face the problem of solving the modular anomaly satisfied by $h_{p,\mu}$
and any its simplification might be important for this crucial step.

Finally, understanding the structure of the anomaly is important also for understanding the origin 
of the modularity from the pure mathematical point of view. While in string theory it comes from
S-duality of type IIB string theory, only recently there was an advance in proving it for some simplest CYs and 
for the simplest case of unit magnetic charge \cite{talkSheshmani} where there is no any modular anomaly
(see also \cite{Feyzbakhsh:2020wvm} for some recent discussion).
It is an exciting challenge to reveal, without referring to the physics of string theory, 
why more general rank 0 DT invariants possess mock modular properties.

\section*{Acknowledgements}

The authors are grateful to the Abdus Salam International Centre for Theoretical Physics 
and the Galileo Galilei Institute for Theoretical Physics for the hospitality 
as well as the INFN for partial support during the completion of this work.

\appendix

\section{Generalized error functions}
\label{ap-generr}

The generalized error functions have been
introduced in \cite{Alexandrov:2016enp,Nazaroglu:2016lmr} (see also \cite{kudla2016theta}).
They are defined by
\be
E_n(\cM;\vu)= \int_{\IR^n} \de \vu' \, e^{-\pi\sum_{i=1}^n(u_i-u'_i)^2} \sign(\cM^{\rm tr} \vu')\, ,
\label{generr-E}
\ee
where $\vu=(u_1,\dots,u_n)$ is $n$-dimensional vector, $\cM$ is $n\times n$ matrix of 
parameters\footnote{The information carried by $\cM$ is in fact highly redundant. 
For example, for $n=1$ the dependence on $\cM$ drops out, whereas
	$E_2$ and $E_3$ depend only on 1 and 3 parameters, respectively.},
and we used the shorthand notation $\sgn(\vu)=\prod_{i=1}^n\sgn(u_i)$.
They generalize the usual error function because at $n=1$ one has $E_1(u)=\Erf(\sqrt{\pi}\, u)$.
Their main role is to provide modular completions of indefinite theta series with quadratic form
of an indefinite signature.

To get kernels of indefinite theta series, we need however functions depending on 
a $d$ dimensional vector with $d\geq n$, rather than $n$-dimensional one,
where $d$ plays the role of the dimension of a lattice and can be arbitrarily large. 
To define such functions, let $\IR^d$ be endowed with a bilinear form $\bfx\cdot\bfy$
of signature $(n_+,n_-)$. Then for any\footnote{In the conventions used to define the generalized error functions 
\cite{Alexandrov:2016enp,Nazaroglu:2016lmr} and which we follow here, the quadratic form appears in theta series with 
an additional minus, $e^{-\pi\I\tau Q(\bfk)}$, so the standard convergent theta series correspond to 
the negative definite signature $(0,d)$.} 
$n\leq n_+$ we take
$\cV$ to be a $d\times n$ matrix which can be viewed as a collection of $n$ vectors, $\cV=(\bfv_1,\dots,\bfv_n)$, 
assumed to span a positive definite subspace in $\IR^d$,
i.e. $\cV^{\rm tr}\cdot\cV$ must be positive definite. We also introduce a $n\times d$ matrix
$\cB$ whose rows define an orthonormal basis for this subspace.
Then we define a boosted version of generalized error functions
\be
\Phi_n^E(\cV;\bfx)=E_n(\cB\cdot \cV;\cB\cdot\, \bfx).
\label{generrPhiME}
\ee
The detailed properties of these functions can be found in \cite{Nazaroglu:2016lmr}.
Most importantly, they do not depend on $\cB$,
at generic\footnote{This specification is important because, as discussed below in \$\ref{subsec-limitPhi}, 
	if the limit is taken in the direction
	where one or several scalar products $\bfv_i\cdot\bfx$ vanish, the result is more subtle.} 
large $\bfx$ reduce to
$\prod_{i=1}^n \sgn (\bfv_i\cdot\bfx)$,
are odd/even for odd/even $n$,
\be 
\Phi_n^E(\cV;-\bfx)=(-1)^n\Phi_n^E(\cV;\bfx),
\label{oddPhi}
\ee 
and annihilated by the Vign\'eras operator 
\be
\label{Vigdif}
V_0\cdot \Phi_n^E(\cV;\bfx)=0,
\qquad
V_\lambda= \p_{\bfx}^2   + 2\pi \( \bfx\cdot \p_{\bfx}  - \lambda\)  .
\ee
The latter property ensures that any linear combination of the boosted generalized error functions
belonging to $L^2(\IR^d)$ gives rise to an indefinite theta series which transforms as 
a modular form \cite{Vigneras:1977}.
Thus, to construct a completion of a holomorphic theta series with a kernel 
given by a linear combination of products of sign functions,
it is sufficient to replace each of these products,
which all have the form $\prod_{i=1}^n \sgn (\bfv_i\cdot\bfx)$
with some vectors $\bfv_i$, by $\Phi_n^E(\{\bfv_i\};\bfx)$ \cite{Alexandrov:2016enp}.

An important role is played also by the operator
\be
\cD(\bfv)=\bfv\cdot\(\bfx+\frac{1}{2\pi}\,\p_\bfx\).
\label{defcD}
\ee
It commutes with the Vign\'eras operator in the following sense
\be
\cD(\bfv) V_\lambda=V_{\lambda+1}\cD(\bfv),
\ee
and thus it allows to construct new modular theta series. 
In particular, one can act by such operators on the generalized error functions
and the result will still define an indefinite theta series transforming as a modular form.

\subsection{The limit of large argument}
\label{subsec-limitPhi}

Here we provide a precise result on the large $\bfx$ limit of the generalized error functions:
\begin{proposition}
	\label{prop-largex}
	\be
	\lim_{\lambda\to\infty}\Phi_{n}^E\(\{\bfv_i\};\lambda\bfx\)=
	\sum_{\cJ\subseteq \Zv_{n}}
	\Phi_{|\cJ|}^E(\{\bfv_k\}_{k\in\cJ};0)\,\prod_{k\in \cJ}\delta_{\bfv_k\cdot\bfx}
	\prod_{i\in \Zv_{n}\setminus \cJ} \sgn (\bfv_i\cdot\bfx),
	\label{large-x-PhiE}
	\ee
	where $\Zv_{n}=\{1,\dots,n\}$ and $|\cJ|$ is the cardinality of the set.
\end{proposition}
\begin{proof}
	First, we note that one can always choose the matrix $\cB$ in the definition \eqref{generrPhiME}
	as $\cB=X^{-1/2} \cV^{\rm tr}$ where $X=\cV^{\rm tr}\cdot \cV$ is a symmetric positive definite matrix
	and hence has a well defined square root.
	Substituting this into \eqref{generrPhiME} and \eqref{generr-E}, one obtains 
	\be 
	\begin{split}
		\Phi_n^E(\cV;\bfx) =&\, 
		\int_{\IR^n} \de \vu' \, e^{-\pi\(X^{-\hf} \cV^{\rm tr}\bfx-\vu'\)^{\rm tr}\(X^{-\hf} \cV^{\rm tr}\bfx-\vu'\)} 
		\sign(X^{\hf} \vu')
		\\
		=&\, 
		(\det X)^{-1/2} \int_{\IR^n} \de \vu \, e^{-\pi \(\cV^{\rm tr}\bfx-\vu\)^{\rm tr} X^{-1}\(\cV^{\rm tr}\bfx-\vu\)} 
		\sign(\vu),
	\end{split}
	\label{intPhiE}
	\ee
	where at the second step we have done the change of the integration variables $\vu'=X^{-\hf} \vu$.	
	
	Then let $\cJ\subseteq \Zv_{n}$ be a subset such that $\bfv_i\cdot \bfx=0$ for $i\in\cJ$ and does not vanish otherwise.
	We can always reorder the indices such that $\cJ=\Zv_{|\cJ|}$. We will label this set by $a,b,\dots$,
	while the remaining indices will be labelled by $\alpha,\beta,\dots$. Accordingly, we can represent
	the matrix $X$ and its inverse in a block diagonal form
	\be 
	X=\(\begin{array}{cc}
		A & B \\ B^{\rm tr} & C	
	\end{array}\),
	\qquad	
	X^{-1}=\(\begin{array}{cc}
		A^{-1} +A^{-1} B D^{-1} B^{\rm tr} A^{-1} & -A^{-1} B D^{-1} \\ -D^{-1} B^{\rm tr} A^{-1} & D^{-1}	
	\end{array}\),	
	\ee 
	where $A_{ab}=\bfv_a\cdot\bfv_b$, $B_{a\alpha}=\bfv_a\cdot\bfv_\alpha$, $C_{\alpha\beta}=\bfv_\alpha\cdot\bfv_\beta$
	and $D=C-B^{\rm tr} A^{-1} B$.
	Using this decomposition in \eqref{intPhiE}, one obtains 
	\bea
	\Phi_n^E(\cV;\lambda\bfx) &=&
	(\det X)^{-1/2}  
	\int_{\IR^{|\cJ|}} \prod_{a}\de u_a  \, 
	e^{-\pi \sum\limits_{a,b}(A^{-1})_{ab}u_a u_b} \prod_{a}\sign(u_a)
	\\
	&& \times 
	\int_{\IR^{n-|\cJ|}} \prod_{\alpha}\de w_\alpha\, 
	e^{-\pi \sum\limits_{\alpha,\beta}(D^{-1})_{\alpha\beta}(w-B^{\rm tr} A^{-1} \vu)_\alpha(w-B^{\rm tr} A^{-1} \vu)_\beta} 
	\prod_{\alpha}\sign(w_\alpha+\lambda\bfv_\alpha\cdot\bfx), 
	\nn 
	\eea
	where $w_\alpha=u_\alpha-\lambda\bfv_\alpha\cdot\bfx$. For large $\lambda$ the last factor given by the product 
	of sign functions can be approximated by 1. As a result, the integral over $w_\alpha$ becomes Gaussian.
	Taking also into account that $\det X=\det A \det D$, one finds that
	\be 
	\begin{split}
		\lim_{\lambda\to\infty}\Phi_{n}^E\(\{\bfv_i\};\lambda\bfx\)
		=&\,(\det A)^{-1/2}  \prod_{\alpha}\sign(\bfv_\alpha\cdot\bfx) 
		\int_{\IR^{|\cJ|}} \prod_{a}\de u_a  \, 
		e^{-\pi \sum\limits_{a,b=1}(A^{-1})_{ab}u_a u_b}
		\\
		=&\, \Phi_{|\cJ|}^E(\{\bfv_a\};0)\prod_{\alpha}\sign(\bfv_\alpha\cdot\bfx).
	\end{split}
	\ee
	Combining together all possible cases corresponding to different subsets $\cJ$, 
	we arrive at the statement of the Proposition.
\end{proof}

\subsection{Identities}
\label{subsec-identPhi}

The generalized error functions satisfy an important identity, which generalizes the one given in 
\cite[Eq.(D.13)]{Alexandrov:2018lgp}:

\begin{proposition}
\label{prop-identPhiE}
\be
\Phi_{n}^E\(\{\bfv_i+\delta_{i,n}\bfv_{n-1}\};\bfx\)+\Phi_n^E\(\{\bfv_i+\delta_{i,n-1}\bfv_{n}\};\bfx\) 
-\Phi_n^E\(\{\bfv_i\};\bfx\) = \Phi_{n-2}^E\(\{\bfv_i\};\bfx\),
\label{identPhiE}
\ee
where for $n=2$ on the r.h.s. we set by definition $\Phi_0^E=1$.
\end{proposition}
\begin{proof}
Let us choose an orthonormal basis $\{\bfb_i\}$ in the positive definite subspace of $\IR^d$ 
such that $\Span(\bfb_1,\dots,\bfb_k)=\Span(\bfv_1,\dots,\bfv_k)$ for all $k=1,\dots, n$.
This is possible due to the assumption that $n$ vectors $\bfv_i$ span the same positive definite subspace
as $\bfb_i$. Then, by construction, $\bfv_i\cdot\bfb_j=0$ for $j>i$.
Taking this property into account, identifying $\bfb_i$ with the rows of the matrix $\cB$, 
and applying the definitions \eqref{generrPhiME} and \eqref{generr-E} to the l.h.s. of \eqref{identPhiE},
one obtains the following expression
\bea
&&
\int_{\IR^n} \de \vu' \, e^{-\pi\sum_{i=1}^n(\bfb_i\cdot \bfx-u'_i)^2} \prod_{i=1}^{n-2} 
\sign\(\bfv_i\cdot \sum_{j=1}^i \bfb_j u'_j\)
\label{proofEn-1}\\
&\times&
\Bigl[\(\sgn(\bfv_{n-1}\cdot \bfu')+\sgn(\bfv_n\cdot \bfu')\)\sgn((\bfv_{n-1}+\bfv_n)\cdot \bfu')
-\sgn(\bfv_{n-1}\cdot \bfu')\sgn(\bfv_n\cdot \bfu')\Bigr],
\nn
\eea
where we denoted $\bfu'=\sum_{j=1}^n \bfb_j u'_j$. Due to the sign identity 

\be
\label{sign-id}
\(\sgn(x_1) + \sgn(x_2) \)\sgn(x_1+x_2)-\sgn(x_1)\sgn(x_2) = 1-\delta_{x_1}\delta_{x_2},
\ee
which is valid with the convention $\sgn(0)=0$, the square bracket can be replaced by 1 because
the term with Kronecker symbols is a measure zero contribution.
In the resulting integral, the integrations over $u'_{n-1}$ and $u'_n$ factorize and 
produce Gaussian integrals both equal to 1. Thus, we remain with
\be 
\int_{\IR^{n-2}} \de \vu' \, e^{-\pi\sum_{i=1}^{n-2}(\bfb_i\cdot \bfx-u'_i)^2} \prod_{i=1}^{n-2} 
\sign\(\bfv_i\cdot \sum_{j=1}^i \bfb_j u'_j\)=\Phi_{n-2}^E\(\{\bfv_i\};\bfx\),
\ee
which proves the Proposition.
\end{proof}

In the case where the vectors $\bfv_i$ correspond to the vectors $\bfv_e$ \eqref{defue} 
defined with respect to an unrooted labelled tree,
the identity \eqref{identPhiE} has a nice geometric interpretation. 
Let us define generalized error functions labelled by such trees
\be 
\Phi^E_{\cT}(\bfx)=\Phi_{|E_\cT|}^E(\{\bfv_e\};\bfx).
\label{defPhiT}
\ee 
Then it is easy to see that the identity \eqref{identPhiE} implies\footnote{The minus sign 
	on the r.h.s. of \eqref{esfT-T123id} appears as follows. 
	First, note that the vectors $\bfv_e$ defined by the trees $\hcT_r$ satisfy 
	$\bfv_{e_1}^{(1)}=\bfv_{e_1}^{(3)}-\bfv_{e_3}^{(3)}$, $\bfv_{e_2}^{(1)}=\bfv_{e_3}^{(3)}$, 
	$\bfv_{e_2}^{(2)}=\bfv_{e_1}^{(3)}$ and $\bfv_{e_3}^{(2)}=\bfv_{e_3}^{(3)}-\bfv_{e_1}^{(3)}$,
	where the upper index indicates the tree with respect to which the vector is defined.
	Comparing to \eqref{identPhiE}, one can identity $\bfv_{n-1}=\bfv_{e_3}^{(3)}$ and  $\bfv_{n}=-\bfv_{e_1}^{(3)}$,
	but then all terms on the l.h.s. of \eqref{identPhiE} and \eqref{PhiT-T123id} have opposite signs.}  
\be
\label{PhiT-T123id}
\Phi^E_{\hat{\cT}_1}(\bfx)+\Phi^E_{\hat{\cT}_2}(\bfx)-\Phi^E_{\hat{\cT}_3}(\bfx)=-\Phi^E_{\cT}(\bfx),
\ee
where the trees $\cT$ and $\hat\cT_r$, $r=1,2,3$, are shown in Fig. \ref{fig-Vign3}.

It is instructive to make the identity \eqref{PhiT-T123id} explicit for trees with 
3 and 5 vertices. The resulting identities will then be used in appendix \ref{ap-rational}.  
To write them, we denote different trees by $\cT_{n,k}$ where $n$ is the number of vertices, while $k$
labels different topologies as in Fig. \ref{fig-Tree3}. 
Different assignments of charges will be denoted by permutations of the standard labelling shown in the same figure,
and we will use notation $\Phi_{\cT_{n,k}}^{(i_1\cdots i_n)}$ for the generalized error function
evaluated on tree $\cT_{n,k}$ where the first vertex of the standard labelling carries charge $\gama_{i_1}$, etc. 
If a vertex carries a sum of charges, it will be labelled by the sum of their indices, like $(i_1+\cdots+ i_\ell)$.
$S_n$ and $C_n$ will denote the groups of arbitrary and cyclic permutations, respectively, of $n$ labels.
To label their elements, we will use the result of their action on the standard labelling: $\sigma=(\sigma(1)\cdots\sigma(n))$.
Finally, $\sigma_{ij}$ will denote a permutation of labels $i$ and $j$, 
and $\iota$ is the inversion $(1\cdots n)\mapsto (n\cdots 1)$. 
We emphasize that in our conventions the inversion and cyclic permutations act on labels, not on their positions. 
For  instance, $\iota(213)=(231)$ and $C_3\ni\tau\,:\, (213)\mapsto(132)$.

With these notations, at $n=3$ \eqref{PhiT-T123id} generates a single identity
which states that the sum over cyclic permutations of the labels on $\cT_{3,1}$ is equal to 1, 
\be
\Phi_{\cT_{3,1}}^{(123)}+\Phi_{\cT_{3,1}}^{(312)}+\Phi_{\cT_{3,1}}^{(231)}=1.
\label{ident-Phi3}
\ee 

\lfig{Standard labelling of unrooted trees with 3 and 5 vertices.}
{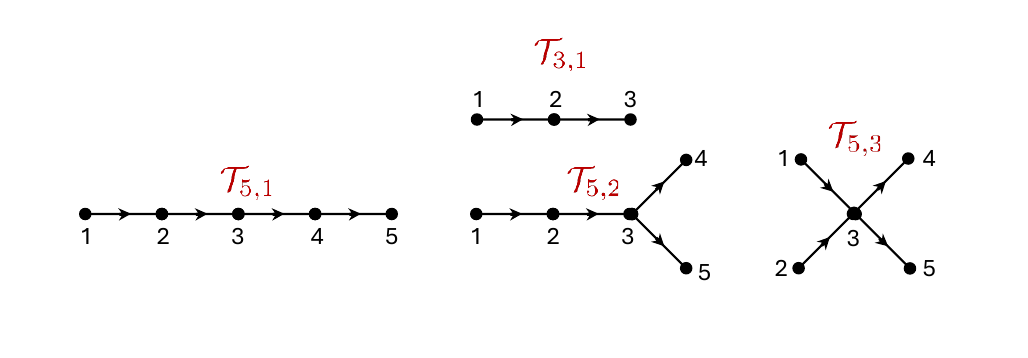}{15.cm}{fig-Tree3}{-1.2cm}

At $n=5$ one obtains already several possibilities which arise from different choices for 
the trees $\cT_r$ in Fig. \ref{fig-Vign3}. Up to permutations, one has three possibilities: 
1) $\cT_1$ coincides with $\cT_{3,1}$ and $\ver_1$ coincides with the third vertex of $\cT_{3,1}$, 
while $\cT_2=\cT_3=\bullet$ are trivial; 
2) $\cT_1$ coincides with $\cT_{3,1}$ and $\ver_1$ coincides with the second (middle) vertex of $\cT_{3,1}$, 
while $\cT_2=\cT_3=\bullet$; 
3) $\cT_1$ and $\cT_2$ coincide with the single tree consisting of two vertices, 
$\cT_{2,1}=\bullet\!\mbox{---}\!\bullet$, 
while $\cT_3=\bullet$.
The corresponding identities can be written as 
\begin{subequations}
\bea
\label{cT3-1}
\Phi_{\cT_{5,1}}^{(12345)}+\Phi_{\cT_{5,1}}^{(12354)}-\Phi_{\cT_{5,2}}^{(12345)}&=&\Phi_{\cT_{3,1}}^{(12(3+4+5))},
\\
\label{cT2-345}
\Phi_{\cT_{5,3}}^{(12345)}-\Phi_{\cT_{5,2}}^{(45321)}-\Phi_{\cT_{5,2}}^{(54321)}&=&\Phi_{\cT_{3,1}}^{(1(3+4+5)2)},
\\
\label{cT2-345-2}
\Phi_{\cT_{5,1}}^{(23451)}+\Phi_{\cT_{5,2}}^{(23541)}+\Phi_{\cT_{5,2}}^{(15342)}&=&\Phi_{\cT_{3,1}}^{(1(3+4+5)2)}.
\eea
\label{ident-PhiT5}
\end{subequations}
To each of these identities one can apply a permutation of indices.
The order of the permutation group $S_5$ is 120. But since, due to symmetries of the trees, 
the first identity is invariant under 
permutation $\sigma_{45}$, the second under $\sigma_{12}$ and $\sigma_{45}$, and the third
under $\sigma_{12}\sigma_{35}$, permutations generate 60, 30 and 60 identities, respectively.
One can check that among these 150 identities only 100 are mutually independent.

Of particular interest are identities that do not involve generalized error functions of lower ranks.
There are two ways to obtain them from \eqref{ident-PhiT5}. The most obvious one is to apply any permutation on labels (345), 
which leaves the r.h.s. invariant, and take a difference.
In this way we get 
\begin{subequations}
\bea
\Phi_{\cT_{5,1}}^{(12345)}+\Phi_{\cT_{5,1}}^{(12354)}-\Phi_{\cT_{5,2}}^{(12345)}
&=&\Phi_{\cT_{5,1}}^{(12435)}+\Phi_{\cT_{5,1}}^{(12453)}-\Phi_{\cT_{5,2}}^{(12435)}
\nn\\
&=&\Phi_{\cT_{5,1}}^{(12543)}+\Phi_{\cT_{5,1}}^{(12534)}-\Phi_{\cT_{5,2}}^{(12543)},
\label{cT5-PhiE-ident1}\\
\Phi_{\cT_{5,3}}^{(12345)}-\Phi_{\cT_{5,2}}^{(45321)}-\Phi_{\cT_{5,2}}^{(54321)}
&=&\Phi_{\cT_{5,3}}^{(12435)}-\Phi_{\cT_{5,2}}^{(35421)}-\Phi_{\cT_{5,2}}^{(53421)}
\nn\\
&=&\Phi_{\cT_{5,3}}^{(12543)}-\Phi_{\cT_{5,2}}^{(43521)}-\Phi_{\cT_{5,2}}^{(34521)}
\nn\\
=\Phi_{\cT_{5,1}}^{(23451)}+\Phi_{\cT_{5,2}}^{(23541)}+\Phi_{\cT_{5,2}}^{(15342)}
&=&
\Phi_{\cT_{5,1}}^{(24351)}+\Phi_{\cT_{5,2}}^{(24531)}+\Phi_{\cT_{5,2}}^{(15432)}
\label{cT5-PhiE-ident2}\\
=\Phi_{\cT_{5,1}}^{(25431)}+\Phi_{\cT_{5,2}}^{(25341)}+\Phi_{\cT_{5,2}}^{(13542)}
&=&\Phi_{\cT_{5,1}}^{(23541)}+\Phi_{\cT_{5,2}}^{(23451)}+\Phi_{\cT_{5,2}}^{(14352)}
\nn\\
=\Phi_{\cT_{5,1}}^{(24531)}+\Phi_{\cT_{5,2}}^{(24351)}+\Phi_{\cT_{5,2}}^{(13452)}
&=& \Phi_{\cT_{5,1}}^{(25341)}+\Phi_{\cT_{5,2}}^{(25431)}+\Phi_{\cT_{5,2}}^{(14532)}.
\nn
\eea
\label{ident-Phi5}
\end{subequations}
The second way is to sum 3 identities \eqref{ident-PhiT5} such that on the r.h.s. one can apply
\eqref{ident-Phi3}. There are many ways to write the resulting identity which are all equivalent due to \eqref{ident-Phi5}.
For example, we have
\be
\begin{split}
	\label{cT5-PhiE-ident123}
	&
	\Phi_{\cT_{5,1}}^{(12345)} + \Phi_{\cT_{5,1}}^{(12354)}+ \Phi_{\cT_{5,1}}^{(21543)} + \Phi_{\cT_{5,1}}^{(21534)}
	\\
	&
	- \Phi_{\cT_{5,2}}^{(12345)}- \Phi_{\cT_{5,2}}^{(21543)} -\Phi_{\cT_{5,2}}^{(45321)}-\Phi_{\cT_{5,2}}^{(54321)}+\Phi_{\cT_{5,3}}^{(12345)}
	=1.
\end{split}
\ee
Of course, to each of the identities \eqref{ident-Phi5}, \eqref{cT5-PhiE-ident123} one can also apply 
an arbitrary permutation.
As a result, there are in total 81 independent identities.

Once the full set of identities is identified, it is possible to show 
that they have a linear combination that is a direct generalization of \eqref{ident-Phi3}:
	\be
	\sum_{\tau\in C_5} \Phi_{\cT_{5,1}}^\tau=1,
	\label{cycle-T51}
	\ee
where the sum goes over cyclic permutations applied to the labelling (12345).
This suggests that a similar identity holds for any odd $n$, namely, that the cyclic permutations of the tree 
of linear topology
$\cT_{n,1}=\bullet\!\mbox{---}\!\bullet\!\mbox{--}\cdots \mbox{--}\!\bullet\!\mbox{---}\!\bullet\,$
with odd number of vertices always sum to one, 
\be 
\sum_{\tau\in C_n}\Phi_{\cT_{n,1}}^\tau=1,   
\qquad \mbox{for odd }n.
\label{ident-Phin-cycl}
\ee
We did not attempt to prove this identity analytically, but we have done extensive numerical checks of 
\eqref{ident-Phin-cycl} and all above identities.

\section{Coefficients $a_\cT$ and $e_\cT$}
\label{ap-coef}

%New commands to draw the trees
\newcommand{\Tsize}{0.25mm}
\newcommand{\scaleAR}{0.9}
\newcommand{\posAR}{0.6}
\newcommand{\drawTa}{	%Tree 1: A straight line of 3 vertices
	\begin{tikzpicture}[
		grow=east,
		every node/.style={circle, fill=black,draw=black, minimum size=4pt, inner sep=0pt},
		]
		
		% Root node
		\node {};	
	\end{tikzpicture}
}
\newcommand{\drawTc}{
	%Tree 1: A straight line of 3 vertices
	\begin{tikzpicture}[
		grow=east,
		every node/.style={circle, fill=black,draw=black, minimum size=4pt, inner sep=0pt},
		edge from parent/.style={draw, line width=\Tsize,postaction={decorate},
			decoration={markings, mark=at position \posAR with {\arrow[scale=\scaleAR]{>}}}},
		level distance=1cm, % Default horizontal spacing
		sibling distance=1cm % Default vertical spacing
		]
		
		% Root node
		\node {}
		child {node {}
			child {node {}
			}
		};	
\end{tikzpicture}
}
\newcommand{\drawTe}[1]{
	\ifcase#1 \or 
	%Tree 1: A straight line of 5 vertices
	\begin{tikzpicture}[scale=0.5,
		grow=east,
		every node/.style={circle, fill=black,draw=black, minimum size=4pt, inner sep=0pt},
		edge from parent/.style={draw, line width=\Tsize,postaction={decorate},
			decoration={markings, mark=at position 0.6 with {\arrow[scale=\scaleAR]{>}}}},
		level distance=1cm, % Default horizontal spacing
		sibling distance=1cm % Default vertical spacing
		]
		
		% Root node
		\node {}
		child {node {}
			child {node {}
				child {node {}
					child {node {}
					}
				}
			}
		};	
	\end{tikzpicture}
	\or 
	%Tree 2: A tree with 3 vertices in a line and a fork with 2 vertices at the end
	\begin{tikzpicture}[scale=0.5,
		grow=east,
		every node/.style={circle, fill=black,draw=black, minimum size=4pt, inner sep=0pt},
		edge from parent/.style={draw, thick, line width=\Tsize,postaction={decorate},
			decoration={markings, mark=at position 0.6 with {\arrow[scale=\scaleAR]{>}}}},
		level distance=1cm, % Default horizontal spacing
		sibling distance=1cm % Default vertical spacing
		]
		
		% Root node
		\node {}
		child {node {}
			child {node {}
				child {node {}} % Increase horizontal spacing
				child {node {}} % Increase horizontal spacing
			}
		};	
	\end{tikzpicture}
	\or
	%Tree 3: A tree in star shape
	\begin{tikzpicture}[scale=0.5,
		grow cyclic,
		every node/.style={circle, fill=black, draw=black, minimum size=4pt, inner sep=0pt},
		edge from parent/.style={
			draw, 
			thick, 
			line width=\Tsize, 
			postaction={decorate},
			decoration={markings, mark=at position 0.6 with {\arrow[scale=\scaleAR]{>}}}
		},
		level distance=1cm,
		sibling angle=90
		]
		\node {}
		child {node {}}
		child {node {}}
		child {node {}}
		child {node {}};
	\end{tikzpicture}
	\fi}
\newcommand{\drawTg}[1]{
	\ifcase#1 \or 
	%Tree 1: A straight line of 7 vertices
	\begin{tikzpicture}[scale=0.5,
		grow=east,
		every node/.style={circle, fill=black,draw=black, minimum size=4pt, inner sep=0pt},
		edge from parent/.style={draw, line width=\Tsize,postaction={decorate},
			decoration={markings, mark=at position 0.6 with {\arrow[scale=\scaleAR]{>}}}},
		level distance=1cm, % Default horizontal spacing
		sibling distance=1cm % Default vertical spacing
		]
		
		% Root node
		\node {}
		child {node {}
			child {node {}
				child {node {}
					child {node {}
						child {node {}
							child{node {} 
							}
						}
					}
				}
			}
		};	
	\end{tikzpicture}
	\or
	%Tree 2: A tree with 5 vertices in a line and a fork with 2 vertices at the end
	\begin{tikzpicture}[scale=0.5,
		grow=east,
		every node/.style={circle, fill=black,draw=black, minimum size=4pt, inner sep=0pt},
		edge from parent/.style={draw, thick, line width=\Tsize,postaction={decorate},
			decoration={markings, mark=at position 0.6 with {\arrow[scale=\scaleAR]{>}}}},
		level distance=1cm, % Default horizontal spacing
		sibling distance=1cm % Default vertical spacing
		]
		
		% Root node
		\node {}
		child {node {}
			child {node {}
				child {node {}
					child {node {}
						child {node {}} % Increase horizontal spacing
						child {node {}} % Increase horizontal spacing
					}
				}
			}
		};	
	\end{tikzpicture}
	\or
	
	%Tree 3: A straight line of 5 edges and a single extra edge that grows downward at the fourth vertex
	\begin{tikzpicture}[scale=0.5,
		grow=east,
		every node/.style={circle, fill=black,draw=black, minimum size=4pt, inner sep=0pt},
		edge from parent/.style={draw, thick, line width=\Tsize,postaction={decorate},
			decoration={markings, mark=at position 0.6 with {\arrow[scale=\scaleAR]{>}}}},
		level distance=1cm, % Default horizontal spacing
		sibling distance=1.5cm % Default vertical spacing
		]
		
		% Root node
		\node {}
		child {node {}
			child {node {}[clockwise from = 0, sibling angle = 90]
				child {node {}
					child {node {}
						child {node {}
						}
					}
				}
				child {node {}}
			}
		};	
	\end{tikzpicture}
	\or
	%Tree 4: A tree with 4 vertices in a line and a fork with 3 vertices at the end
	\begin{tikzpicture}[scale=0.5,
		grow=east,
		every node/.style={circle, fill=black, draw=black, minimum size=4pt, inner sep=0pt},
		edge from parent/.style={
			draw, 
			thick, 
			line width=\Tsize, 
			postaction={decorate},
			decoration={markings, mark=at position 0.6 with {\arrow[scale=\scaleAR]{>}}}
		},
		level distance=1cm,
		sibling distance=1 cm
		]
		\node {}
		child {node {}
			child {node {}
				child {node {}
					child {node {}}
					child {node {}}
					child {node {}}
				}
			}
		};
	\end{tikzpicture}
	\or
	%Tree 5: A straight line of 4 edges, an extra edge that grows downward at the third vertex and a 2 edges fork at the last vertex
	\begin{tikzpicture}[scale=0.5,
		grow=east,
		every node/.style={circle, fill=black,draw=black, minimum size=4pt, inner sep=0pt},
		edge from parent/.style={draw, thick, line width=\Tsize,postaction={decorate},
			decoration={markings, mark=at position 0.6 with {\arrow[scale=\scaleAR]{>}}}},
		level distance=1cm, % Default horizontal spacing
		sibling distance=1cm % Default vertical spacing
		]
		
		% Root node
		\node {}
		child {node {}
			child {node {}[clockwise from = 0, sibling angle = 90,level distance =0.7cm]
				child {node {}[clockwise from = 30, sibling angle = 60,level distance =1cm]
					child {node {}}
					child {node {}}
				}
				child {node {}}
			}
		};	
	\end{tikzpicture}
	\or
	%Tree 6: A tree with 4 vertices in a line and a two 2-fork one at the beginning and one at the end.
	\begin{tikzpicture}[scale=0.5,
		grow=east,
		every node/.style={circle, fill=black,draw=black, minimum size=4pt, inner sep=0pt},
		edge from parent/.style={draw, thick, line width=\Tsize,postaction={decorate},
			decoration={markings, mark=at position 0.6 with {\arrow[scale=\scaleAR]{.}}}},
		level distance=1cm, % Default horizontal spacing
		sibling distance=1cm % Default vertical spacing
		]
		
		% Root node
		\node {} [clockwise from = 0, sibling angle=120]
		child {node {}[clockwise from = 0, sibling angle=150]
			child {node {}[decoration={markings, mark=at position 0.6 with {\arrow[scale=\scaleAR]{>}}}]
				child {node {}[clockwise from = 30, sibling angle=60]
					child {node {}} 
					child {node {}} 
				}
			}
			[decoration={markings, mark=at position 0.6 with {\arrow[scale=\scaleAR]{>}}}]}
		child[clockwise from = -30]{node{}[decoration={markings, mark=at position 0.6 with {\arrow[scale=\scaleAR]{<}}}]}
		child[clockwise from = 30]{node{}[decoration={markings, mark=at position 0.6 with {\arrow[scale=\scaleAR]{<}}}]};	
	\end{tikzpicture}
	\or
	%Tree 7: A tree with 3 vertices in a line and a fork with 4 vertices at the end
	\begin{tikzpicture}[scale=0.5,
		grow=east,
		every node/.style={circle, fill=black, draw=black, minimum size=4pt, inner sep=0pt},
		edge from parent/.style={
			draw, 
			thick, 
			line width=\Tsize, 
			postaction={decorate},
			decoration={markings, mark=at position 0.6 with {\arrow[scale=\scaleAR]{>}}}
		},
		level distance=1cm,
		sibling distance=0.7cm
		]
		\node {}
		child {node {}
			child {node {}
				child {node {}}
				child {node {}}
				child {node {}}
				child {node {}}
			}
		};
	\end{tikzpicture}
	\or
	%Tree 8: A tree in star shape
	\begin{tikzpicture}[scale=0.5,
		grow cyclic,
		every node/.style={circle, fill=black, draw=black, minimum size=4pt, inner sep=0pt},
		edge from parent/.style={
			draw, 
			thick, 
			line width=\Tsize, 
			postaction={decorate},
			decoration={markings, mark=at position 0.6 with {\arrow[scale=\scaleAR]{>}}}
		},
		level distance=1cm,
		sibling angle=60
		]
		\node {}
		child {node {}}
		child {node {}}
		child {node {}}
		child {node {}}
		child {node {}}
		child {node {}};
	\end{tikzpicture}
	\or
	%Tree 9: A tree with 3 vertices in a line and two forks from the third vertex of length two (2 edges) each
	\begin{tikzpicture}[scale=0.5,
		grow=east,
		every node/.style={circle, fill=black,draw=black, minimum size=4pt, inner sep=0pt},
		edge from parent/.style={draw, thick, line width=\Tsize,postaction={decorate},
			decoration={markings, mark=at position 0.6 with {\arrow[scale=\scaleAR]{>}}}},
		level distance=1cm, % Default horizontal spacing
		sibling distance=1cm % Default vertical spacing
		]
		
		% Root node
		\node {}
		child {node {}
			child {node {}
				child {node {} 
					child{node {}}
				} 
				child {node {}
					child {node {}
				}}
			}
		};	
	\end{tikzpicture}
	\or
	%Tree 10: A tree with 2 vertices in a line and a two forks one at the beginning with 2 vertices and one at the end with 3 vertices.
	\begin{tikzpicture}[scale=0.5,
		grow=east,
		every node/.style={circle, fill=black,draw=black, minimum size=4pt, inner sep=0pt},
		edge from parent/.style={draw, thick, line width=\Tsize,postaction={decorate},
			decoration={markings, mark=at position 0.6 with {\arrow[scale=0.7]{.}}}},
		level distance=1cm, % Default horizontal spacing
		sibling distance=1cm % Default vertical spacing
		]
		
		% Root node
		\node {} [clockwise from = 0, sibling angle=120]
		child {node {}[clockwise from = 0, sibling angle=150][decoration={markings, mark=at position 0.6 with {\arrow[scale=\scaleAR]{>}}}][clockwise from = 45, sibling angle=45]
			child {node[minimum size = 4pt] {}} 
			child {node {}} 
			child {node {}}
			[decoration={markings, mark=at position 0.6 with {\arrow[scale=\scaleAR]{>}}},sibling angle=150]}
		child[clockwise from = -30]{node{}[decoration={markings, mark=at position 0.6 with {\arrow[scale=0.7]{<}}}]}
		child[clockwise from = 30]{node{}[decoration={markings, mark=at position 0.6 with {\arrow[scale=0.7]{<}}}]};	
	\end{tikzpicture}
	\or
	%Tree 11: A straight line of 5 edges and a two extra edges that grows downward and upward at the third vertex
	\begin{tikzpicture}[scale=0.5,
		grow=east,
		every node/.style={circle, fill=black,draw=black, minimum size=4pt, inner sep=0pt},
		edge from parent/.style={draw, thick, line width=\Tsize,postaction={decorate},
			decoration={markings, mark=at position 0.6 with {\arrow[scale=\scaleAR]{>}}}},
		level distance=1cm, % Default horizontal spacing
		sibling distance=1.5cm % Default vertical spacing
		]
		
		% Root node
		\node {}
		child {node {}
			child {node {}[clockwise from = 90, sibling angle = 90]
				child {node {}}
				child {node {}[clockwise from = 0, sibling angle = 90]
					child {node {}
					}
				}
				child {node {}}
			}
		};	
	\end{tikzpicture}
	\fi}

\newcommand\tableT[1][7]{\rule{0pt}{#1ex}} % Top strut
\newcommand\tableB[1][-3]{\rule[#1ex]{0pt}{0pt}} % Bottom strut

\begin{tabular}{|c|c|c|c|}
	\hline
	%%%%%%%%%%%%%%%%% The header of the table
	Name                 										& 			Tree       																				& 			$a_\cT$                                             					& $e_\cT$                                               \\ \hline
	%%%%%%%%%%%%%%%%% End of the header of the table
	%Row of the trivial tree with one vertex:
	$\cT_{1,1}$         \tableT[3] \tableB[-1.5]							& 			\drawTa       																				& 			$1$                                                 					& 		$1$                                        		\\ \hline
	%Row of the linear tree with n=3 vertices:
	$\cT_{3,1}$         \tableT[3] \tableB[-1.5]   								&      		\drawTc																					& 			\small$-\frac13$                                    					& 		\small$\frac13$                              	\\ \hline
	%Row of the linear tree with n=5 vertices:
	$\cT_{5,1}$      \tableT[3] \tableB[-1.5]   								& 			\drawTe{1}  																				& 			\small$\frac{2}{15}$                                					& 		\small$\frac15$                             	\\ \hline
	%Row of the fork tree with n=5 vertices:
	\makecell{\vspace{0.5em}$\cT_{5,2}$} \tableT[4]  			& 			\makecell{\vspace{0.1em} \drawTe{2}}  														& 			\makecell{\vspace{0.5em}$-\frac{1}{15}$}     							& 		\makecell{\vspace{0.5em}$\frac{1}{15}$}        	\\ \hline
	%Row of the star tree with n=5 vertices:
	\makecell{\vspace{0.5em}$\cT_{5,3}$} \tableT[5]  			& 			\makecell{\vspace{0.1em}\drawTe{3}}  														& 			\makecell{\vspace{0.5em}$\frac15$}           							& 		\makecell{\vspace{0.5em}$\frac{7}{15}$}        	\\ \hline
	%Row of the linear tree with n=7 vertices T_{7,1}:
	$\cT_{7,1}$         \tableT[3] \tableB[-1.5]							& 			\makecell{\vspace{0.1em}\drawTg{1}}  												& 			\makecell{\vspace{0.3em}$-\frac{17}{315}$}                              & 		\makecell{\vspace{0.3em}$\frac17$}         		\\ \hline
	%Row of the linear tree with n=7 vertices T_{7,2}:
	\makecell{\vspace{0.7em}$\cT_{7,2}$} \tableT[5]  			& 			\makecell{\vspace{0.3em}\drawTg{2}}  														& 			\makecell{\vspace{0.7em}$\frac{8}{315}$}     							& 		\makecell{\vspace{0.7em}$\frac{3}{35}$}        	\\ \hline
	%Row of the linear tree with n=7 vertices T_{7,3}:
	\makecell{\vspace{0.9em}$\cT_{7,3}$}\tableT[5]				&			\makecell{\vspace{0.1em}\drawTg{3}}															&			\makecell{$-\frac{1}{35}$\vspace{0.9em}}								&		\makecell{$\frac{1}{35}$\vspace{0.9em}}        	\\ \hline
	%Row of the linear tree with n=7 vertices T_{7,4}:
	$\cT_{7,4}$ \tableT[5]  									& 			\makecell{\vspace{0.1em}\drawTg{4}}  														& 			$\frac{8}{105}$     													& 		$-\frac17$    					  				\\ \hline
\end{tabular}
\hspace{2cm}
\begin{tabular}{|c|c|c|c|}
		\hline
	%%%%%%%%%%%%%%%%% The header of the table
	Name                 										& 			Tree      																				& 			$a_\cT$                                             					& $e_\cT$                                               \\ \hline
		%Row of the linear tree with n=7 vertices T_{7,5}:
	\makecell{$\cT_{7,5}$\vspace{0.8em}} \tableT[5]  			& 			\makecell{\drawTg{5}\vspace{0.1em}}  														& 			\makecell{$-\frac{1}{105}$\vspace{0.8em}}   	 						& 		\makecell{$-\frac{13}{105}$\vspace{0.8em}}     	\\ \hline
		%Row of the linear tree with n=7 vertices T_{7,6}:
	\makecell{$\cT_{7,6}$\vspace{0.7em}} \tableT[5]  			& 			\makecell{\drawTg{6}}  																		& 			\makecell{$-\frac{1}{63}$\vspace{0.7em}}     							& 		\makecell{$\frac{11}{105}$\vspace{0.7em}}		\\ \hline
		%Row of the linear tree with n=7 vertices T_{7,7}:
	\makecell{$\cT_{7,7}$\vspace{0.7em}} \tableT[6]				& 			\makecell{\drawTg{7}}  																		& 			\makecell{$-\frac{1}{35}$\vspace{0.7em}}     							& 		\makecell{$-\frac{53}{105}$\vspace{0.7em}}     	\\ \hline
	%Row of the linear tree with n=7 vertices T_{7,8}:
	\makecell{$\cT_{7,8}$\vspace{0.9em}}\tableT[6]				& 			\makecell{\drawTg{8}}  																		& 			\makecell{$\frac17$\vspace{0.9em}}           							& 		\makecell{$-\frac{31}{21}$\vspace{0.9em}}      	\\ \hline
	%Row of the linear tree with n=7 vertices T_{7,9}:
	\makecell{$\cT_{7,9}$\vspace{0.5em}} \tableT[4]  			& 			\makecell{\drawTg{9}}  
	& 			\makecell{$-\frac{2}{35}$\vspace{0.5em}}     							& 		\makecell{$-\frac{1}{105}$\vspace{0.5em}}      	\\ \hline
	%Row of the linear tree with n=7 vertices T_{7,10}:
	\makecell{$\cT_{7,10}$\vspace{0.9em}} \tableT[5] 			& 			\makecell{\drawTg{10}} 
	& 			\makecell{$-\frac{1}{21}$\vspace{0.9em}}     							& 		\makecell{$\frac{19}{105}$\vspace{0.9em}}      	\\ \hline
	%Row of the linear tree with n=7 vertices T_{7,11}:
	\makecell{$\cT_{7,11}$\vspace{0.9em}} \tableT[5.5] 			& 			\makecell{\drawTg{11}} 
	& 			\makecell{$\frac{2}{105}$\vspace{0.9em}}     							& 		\makecell{$-\frac{23}{105}$\vspace{0.9em}}     	\\ \hline
\end{tabular}

\subsection*{Proof of the identity \eqref{eT-T123id}}

To prove the identity \eqref{eT-T123id}, we will rely on a similar identity satisfied
by the coefficients $a_\cT$ \cite{Alexandrov:2018lgp}:
\be
\label{aT-T123id}
a_{\hat{\cT}_1}+a_{\hat{\cT}_2}-a_{\hat{\cT}_3}=a_{\cT_1}a_{\cT_2}a_{\cT_3},
\ee
where the trees $\cT_r$ and $\hat \cT_r$ are drawn in Fig. \ref{fig-Vign3}. 
We will proceed by induction. 

At $n=3$ there is a single identity of the form \eqref{eT-T123id}
corresponding to the trivial trees $\cT_r$ consisting of one vertex.
Then it is easy to check that the identity holds for the values given in the above table.
(One should take into account that under a flip of the orientation of an edge, the coefficient flips the sign.)

Then let us assume that the identity \eqref{eT-T123id} holds for all trees $\hcT_r$ with $n'<n$ vertices. 
To treat the case of $n$ vertices, we use the formula \eqref{res-eT} in the l.h.s. of \eqref{eT-T123id} to get
\be
-\sum_{\smash{\mathop{\cup}\cT'_k\simeq\hcT_1 }}\,e_{\hcT_1/\{\cT'_k\}}\prod_{k=1}^m a_{\cT'_k}
-\sum_{\smash{\mathop{\cup}\cT'_k \simeq\hcT_2 }}\,e_{\hcT_2/\{\cT'_k\}}\prod_{k=1}^m a_{\cT'_k}
+\sum_{\smash{\mathop{\cup}\cT'_k \simeq\hcT_3 }}\,e_{\hcT_3/\{\cT'_k\}}\prod_{k=1}^m a_{\cT'_k},
\label{lhs-eT}
\ee
where we recall that in each decomposition there should be at least one non-trivial tree (having more than one vertex).
For each tree $\hcT_r$, its decomposition can be of three types:
\begin{enumerate}
	\item both edges $e_i$ belong to one of the trees of the decomposition;
	\item only one of the two edges $e_i$ belongs to a tree of the decomposition;
	\item neither of the edges $e_i$ belongs to the decomposition trees.
\end{enumerate}
We will show below that the contributions of type 2 cancel each other, the contributions of type 3 are 
canceled by contributions of type 1, while the remaining contributions of type 1
give exactly $-e_\cT$. 

We start with the type 2 contributions and consider a decomposition $\mathop{\cup}\cT'_k$ of $\hcT_1$ such that 
$e_2\subset\cT'_{k_0}$. It is clear that $\cT'_{k_0}\subset \cT_1\mathop{\cup}_{e_2}\cT_3$, i.e. it is a subtree of 
the tree formed by $\cT_1$ and $\cT_3$ joint by $e_2$, while all other trees $\cT'_k$, $k\neq k_0$, 
are subtrees of $\cT_1$, $\cT_2$ or $\cT_3$. This implies that the same set of subtrees provides a decomposition of $\hcT_2$. 
Extracting the contribution of these two decompositions from \eqref{lhs-eT}, one finds 
\be
-\(e_{\hcT_1/\{\cT'_k\}}+e_{\hcT_2/\{\cT'_k\}}\)\prod_{k=1}^m a_{\cT'_k}.
\label{type2-eT}
\ee
After collapsing all subtrees $\cT'_k$, the edge $e_2$ disappears, while $e_1\subset\hcT_1$ and $\e_3\subset\hcT_2$ remain.
But now they connect the same vertices and differ only by orientation. 
Therefore, $e_{\hcT_1/\{\cT'_k\}}=-e_{\hcT_2/\{\cT'_k\}}$ and the combination \eqref{type2-eT} vanishes.
In exactly the same way one can show the vanishing of all other type 2 contributions. 

Next we consider a type 3 decomposition. In this case each subtree $\cT'_k$ also belongs to one of the trees $\cT_s$.
Therefore, we can split them into 3 subsets $\{\cT'_{s,l}\}$ such that $\cup_l \cT'_{s,l}\simeq\cT_s$.
It is also clear that $\cup_{s,l}\cT'_{s,l}\simeq\hcT_r$ for all $r$.
Therefore, the contribution of such decomposition into \eqref{lhs-eT} is given by 
\be
-\(e_{\hcT_1/\{\cT'_{s,l}\}}+e_{\hcT_2/\{\cT'_{s,l}\}}-e_{\hcT_3/\{\cT'_{s,l}\}}\)\prod_{s=1}^3 \prod_{l=1}^{m_s} a_{\cT'_{s,l}}.
\label{type3-eT}
\ee
The crucial observation is that the trees $\hcT_r/\{\cT'_{s,l}\}$ are constructed from the three trees $\cT_s/\{\cT'_{s,l}\}$ 
exactly as in Fig. \ref{fig-Vign3}. Since the number of their vertices is less than $n$ (because at least one of $\cT'_{s,l}$
must be non-trivial), they are subject to the induction hypothesis and the contribution \eqref{type3-eT} is equal to
\be
e_{\cT/\{\cT'_{s,l}\}}\prod_{s=1}^3 \prod_{l=1}^{m_s} a_{\cT'_{s,l}}.
\label{type2-eT2}
\ee

Let us now construct a type 1 decomposition that is associated with the type 3 decomposition $\{\cT'_{s,l}\}$ just analyzed. 
Denote $l_s$ the index such that $\ver_s\in \cT'_{s,l_s}$. Using the three trees $\cT'_{s,l_s}$
in place of $\cT_s$, one can construct the three trees $\hcT'_r=\cup_{e_i}\cT'_{s,l_s}$ as in Fig. \ref{fig-Vign3}.
Then $\{\hcT'_r,\cT'_{s,l}\}_{s,l\ne l_s}$ provide decompositions of $\hcT_r$ of type 1.
Their contribution to \eqref{lhs-eT} is given by 
\be
\label{type1-eT}
- e_{\cT'/\{\cT'_{s,l}\}_{s,l\ne l_s}} \(a_{\hcT'_1}+a_{\hcT'_2}-a_{\hcT'_3}\)
\prod_{s=1}^3 \prod_{l\ne l_s} a_{\cT'_{s,l}},
\ee
where $\cT'=\hcT_r/\hcT'_r$ is independent of $r$. Furthermore, $\cT'=\cT/\{\cT'_{s,l_s}\}$ so that 
the $e$-coefficients in \eqref{type2-eT2} and \eqref{type1-eT} coincide.
On the other hand, applying \eqref{aT-T123id} to the expression in brackets in \eqref{type1-eT} one gets the missing factor 
$\prod_{s=1}^3 a_{\cT'_{s,l_s}}$. As a result, the two contributions \eqref{type2-eT2} and \eqref{type1-eT} cancel each other.

However, for each $\hcT_r$ there is one type 1 decomposition that was not accounted for by the above construction.
It consists of a single non-trivial tree obtained by joining the two edges $e_i$ and a set of trivial subtrees, 
i.e. given by a single vertex. Such a decomposition can be obtained from
a type 3 decomposition where all $\{\cT'_{s,l}\}$ are trivial, but as we know, it is excluded from 
the sum in \eqref{lhs-eT}. 
Thus, we remain with the following contribution
\be
-e_{\cT}\(a_{(e_1,e_2)}+a_{(e_2,e_3)}-a_{(e_1,e_3)}\) = -e_{\cT}\(\frac13 +\frac13+\frac13\)=-\e_\cT,
\ee
which completes the proof.

\section{Irrational vs. rational coefficients}
\label{ap-rational}

In this appendix we verify Conjecture \ref{conj-eT} for a few low values of $n$.
We will extensively use notations $\esf_{\cT_{n,k}}^{(i_1\dots i_n)}=\Phi_{\cT_{n,k}}^{(i_1\dots i_n)}(0)$
and $\frK_{\cT_{n,k}}^{(i_1\dots i_n)}$,
similar to the ones introduced in \$\ref{subsec-identPhi}.
See also the paragraph above \eqref{ident-Phi3} for our conventions on permutations acting on labels of these objects.

The first non-trivial case happens at $n=3$. The function \eqref{exprcEf-new} then reads
\be 
\Ef_3=\frac16\,\sum_{\rm cyclic\ perm.} \gamma_{12}\gamma_{23}\(\sgn(\gamma_{1,2+3})\,\sgn(\gamma_{1+2,3})
+\delta_{123} \esf_{\cT_{3,1}}^{(123)}\),
\label{proof-n=3}
\ee 
where we denoted $\delta_{123}=\delta_{\gamma_{1,2+3}}\delta_{\gamma_{1+2,3}}$, 
while $\cT_{3,1}$ is the single unrooted tree with 3 vertices 
as shown in the table of appendix \ref{ap-coef}.
It is easy to check that $\gamma_{12}\gamma_{23}\delta_{123}$ is invariant under cyclic permutations.
Therefore, one can apply the identity \eqref{ident-eT}, 
which in this case is a specification of \eqref{ident-Phi3} to $\bfx=0$ and
allows to replace the coefficients $\esf_{\cT_{3,1}}$ by 1/3.
This results in the representation \eqref{En0new2} with $e_{\cT_{3,1}}=1/3$
in agreement with the Conjecture and the value given in appendix \ref{ap-coef}.

For $n=4$ the calculation is almost the same as at this order there are no new non-trivial coefficients $e_\cT$.
In this case it is enough to consider the contribution to $\Ef_4$ where two of three $\Gamma_e$'s vanish.
One can show that it is given by
\be 
\frac{1}{4\cdot 4!}\sum_{\rm perm.} \Biggl[\delta_{12(3+4)}\gamma_{34}\sgn(\gamma_{1+2+3,4}) 
\sum_{\rm cyclic\ perm. \atop of\ 1,2,3+4}\gamma_{12}\gamma_{2,3+4} \esf_{\cT_{3,1}}^{(12(3+4))}\Biggr].
\ee
The second sum is the same as in \eqref{proof-n=3} so that the coefficients $\esf_{\cT_{3,1}}$
can again be replaced by 1/3 consistently with the Conjecture.

For $n=5$ there are three types of trees contributing to \eqref{exprcEf-new} and shown in Figure \ref{fig-Tree3}.
It is convenient to rewrite the sum over different labellings of each tree as a sum over all 
permutations of labels divided by the symmetry factor of the tree:
\be
\Ef_5=	\frac{1}{5!}
	\sum_{\sigma \in \cS_5}
	\sigma\[\frac12\, \frK_{\cT_{5,1}}\,\Ssf_{\cT_{5,1}}
	+\frac12\,\frK_{\cT_{5,2}}\,\Ssf_{\cT_{5,2}}
	+\frac{1}{24}\, \frK_{\cT_{5,3}}\,\Ssf_{\cT_{5,3}}\].
\label{En0-n5-perm}
\ee
Then to prove the Conjecture, we need to consider two cases where each tree has either two or all four edges 
with vanishing $\Gamma_e$.
(For an odd number of vanishing $\Gamma_e$'s the corresponding coefficients $\esf_\cT$ vanish and
there is nothing to prove.) We will start with the first case which is simpler and goes along the same lines 
as for $n=3$ and 4. 

The contribution to $\Ssf_{\cT_{5,k}}$ corresponding to two vanishing $\Gamma_e$'s can be written as
\be
\sum_{e_1,e_2\in E_{\cT_{5,k}}}
\esf_{\cT_{\{e_1,e_2\}}} \delta_{\Gamma_{e_1}}\delta_{\Gamma_{e_2}} \sgn(\Gamma_{e'_1})\sgn(\Gamma_{e'_2}),
\label{En0-m0-2del}
\ee
where $\{e'_1,e'_2\}=E_\cT \backslash \{e_1,e_2\}$ and
$\cT_{\{e_1,e_2\}}$ is the tree obtained from $\cT_{5,k}$ by contracting the edges $e'_1$ and $e'_2$,
which topologically always coincides with $\cT_{3,1}$.
Then using permutations and the conditions implied by the Kronecker symbols, one can show that 
in the corresponding contribution to \eqref{En0-n5-perm} all terms can be made proportional just 
to the following two combinations
\be
\begin{split}
\Delta_1^{(12345)} = &\,
\delta_{\gamma_{1,2+3+4+5}}\delta_{\gamma_{2,1+3+4+5}} \sgn(\gamma_{3,1+2+4+5})\sgn(\gamma_{5,1+2+3+4}),
\\
\Delta_2^{(12345)} =&\,
\delta_{\gamma_{1+2+3,4+5}}\delta_{\gamma_{1,2+3+4+5}}\sgn(\gamma_{2,1+3+4+5}) \sgn(\gamma_{5,1+2+3+4}).
\end{split}
\label{defDel}
\ee
The resulting expression reads
\be
\label{2del-step1}
\begin{split}
	-\frac{1}{2\cdot 5!}& 
	\sum_{\sigma \in S_5}
	\sigma \biggl[ 
	\Delta_1^{(12345)}\, 
	\( \(\frK^{(12345)}_{\cT_{5,1}}+\frK^{(34521)}_{\cT_{5,1}}-\frK_{\cT_{5,2}}^{(12435)}\) \esf_{\cT_{3,1}}^{(12(3+4+5))}
	\right.
	\\
	&\left.\qquad
	+\(\frK^{(15432)}_{\cT_{5,1}}+\frK_{\cT_{5,2}}^{(15432)}+\frK_{\cT_{5,2}}^{(13425)}-\frK_{\cT_{5,2}}^{(54321)}
    +\frK_{\cT_{5,3}}^{(12435)}\)\esf_{\cT_{3,1}}^{(1(3+4+5)2)}\)
	\\
	&
	+\Delta_2^{(12345)}
	\(2\(\frK^{(12345)}_{\cT_{5,1}}+\frK_{\cT_{5,2}}^{(54321)}\)\esf_{\cT_{3,1}}^{1(2+3)(4+5)} 
	+\frK_{\cT_{5,1}}^{(23145)}\esf_{\cT_{3,1}}^{(2+3)1(4+5)}\)\biggr].
\end{split}
\ee
Note that the factors \eqref{defDel} possess the following symmetry properties:
$\Delta_1^{(12345)}$ is invariant under permutations $\sigma_{12}$  and $\sigma_{35}$, while
$\Delta_2^{(12345)}$ is antisymmetric under $\sigma_{23}$ and $\sigma_{45}$. 
Using these properties again together with 
the conditions implied by the Kronecker symbols, one can recombine the factors of Dirac products of charges
defined in \eqref{defGame} to rewrite \eqref{2del-step1} in the following form
\be
\begin{split}
	-\frac{1}{4\cdot 5!}&
	\sum_{\sigma \in S_5}
	\sigma \biggl[ \biggl(
	 \Delta_1^{(12345)}\gamma_{12}^2\gamma_{34}\gamma_{45}\, 
	\( \esf_{\cT_{3,1}}^{(12(3+4+5))} 
	+\esf_{\cT_{3,1}}^{(1(3+4+5)2)} + \esf_{\cT_{3,1}}^{((3+4+5)12)}\)
	\\
	&\qquad
	+\hf\,\Delta_2^{(12345)}\gamma_{1,4+5}^2\gamma_{23}\gamma_{45} 
	\(\,\esf_{\cT_{3,1}}^{1(2+3)(4+5)}
	+ \esf_{\cT_{3,1}}^{(2+3)1(4+5)}+\esf_{\cT_{3,1}}^{(2+3)(4+5)1} \)
	\biggr)\biggr].
\end{split}
\ee
This allows to apply the identity \eqref{ident-eT} and replace the coefficients $\esf_{\cT_{3,1}}$ 
by 1/3, as required by the Conjecture.

The second case to be analyzed is where all edges have vanishing $\Gamma_e$. 
Then $\Ssf_{\cT_{5,k}}$, defined in \eqref{exprSsf}, reduces to
$\esf_{\cT_{5,k}}\delta_{12345}$
where $\delta_{12345}=\prod_{e\in E_{\cT_{5,k}}}\delta_{\Gamma_e}$ and, due to proposition \ref{prop-vanishG}, 
it is independent both of the topology of the tree (labelled by $k$) and of its labelling. 
Thus, the corresponding contribution to \eqref{En0-n5-perm} is given by
\be
\label{En5-m0-first}
	\frac{1}{2\cdot 5!}\,
	\delta_{12345} \sum_{\sigma \in S_5}
	\sigma\[\frK_{\cT_{5,1}}\,\esf_{\cT_{5,1}}
	+\frK_{\cT_{5,2}}\,\esf_{\cT_{5,2}}
	+\frac{1}{12}\,\frK_{\cT_{5,3}}\,\esf_{\cT_{5,3}} \].
\ee
The idea, which allows to prove the Conjecture, i.e. to replace all $\esf_{\cT_{5,k}}$ by some rational numbers,
is to recombine these coefficients so that one could apply identities similar to the cyclic identity that was used above.
To this end, we use the fact that, given the conditions implied by the Kronecker symbols,
the factors $\frK_{\cT_{5,k}}$ also satisfy numerous relations. We apply them to reduce the number of appearing structures
and this will automatically produce the combinations of $\esf_{\cT_{5,k}}$ subject to the mentioned identities.

We start by noticing the following two relations
\begin{subequations}
\bea
\frK_{\cT_{5,3}}^{(12345)} +\frK_{\cT_{5,2}}^{(12345)}+\frK_{\cT_{5,2}}^{(14352)}+\frK_{\cT_{5,2}}^{(15324)}
&\stackrel{\delta}{=}& 0,
\label{id-frK-StarFork}
\\
\frK_{\cT_{5,1}}^{(12345)}-\frK_{\cT_{5,1}}^{(51234)} - \frK_{\cT_{5,2}}^{(43215)}+\frK_{\cT_{5,2}}^{(12345)} 
&\stackrel{\delta}{=} & 0,
\label{id-frK-LinFork}
\eea
\end{subequations}
where the sign $\stackrel{\delta}{=}$ indicates that they hold in the presence of $\delta_{12345}$.
The first one can be used in \eqref{En5-m0-first} to exclude $\frK_{\cT_{5,3}}$. 
Moreover, since the permutations appearing in the last three terms in \eqref{id-frK-StarFork} leave 
$\esf_{\cT_{5,3}}^{(12345)}$ invariant, we have
\be
\frac{1}{12}\sum_{\sigma\in S_5}\sigma\Bigl[\frK_{\cT_{5,3}}\,\esf_{\cT_{5,3}} \Bigr]
\stackrel{\delta}{=}
-\frac{1}{4}\sum_{\sigma\in S_5}\sigma\Bigl[\frK_{\cT_{5,2}}\,\esf_{\cT_{5,3}}\Bigr].
\label{rel-frK-32}
\ee
The second relation \eqref{id-frK-LinFork} is used to produce a sum over cyclic permutations of $\esf_{\cT_{5,1}}$.
Applying it four times on $\frK_{\cT_{5,1}}^{(12345)}$ and its cyclic permutations, one finds 
\be
\label{cycAvg-frK51}
\begin{split}
	\frK_{\cT_{5,1}}^{(12345)} \stackrel{\delta}{=}&\, \frac{1}{5}\sum_{\tau \in \cC_5} \tau\[\frK_{\cT_{5,1}}^{(12345)}\]
	-\frac{1}{5}\[2\frK_{\cT_{5,2}}^{(12345)} +2\frK_{\cT_{5,2}}^{(54321)} 
	+ \frK_{\cT_{5,2}}^{(51234)}+ \frK_{\cT_{5,2}}^{(15432)}
	\right.
	\\
	&\left.\qquad\qquad
	- \frK_{\cT_{5,2}}^{(34512)}-\frK_{\cT_{5,2}}^{(32154)}- 2\frK_{\cT_{5,2}}^{(23451)}-2\frK_{\cT_{5,2}}^{(43215)} \].
\end{split}
\ee
Substituting \eqref{rel-frK-32} and \eqref{cycAvg-frK51} into \eqref{En5-m0-first} 
and shifting the sum over permutations so that to bring the labels of each factor $\frK_{\cT_{5,k}}$ to
the standard form (12345), one obtains 
\be
\frac{1}{2\cdot 5!}\,
\delta_{12345} \sum_{\sigma \in S_5}
\sigma\[\frac{1}{5}\, \frK_{\cT_{5,1}}\sum_{\tau\in C_5}\esf_{5,1}^\tau
+\frK_{\cT_{5,2}}\,\check\esf_{\cT_{5,2}}\],
\label{En5-m0-first2}
\ee
where we denoted
\be
\label{esf5b}
\begin{split}
\check\esf_{\cT_{5,2}}
&=
\esf_{\cT_{\cT_{5,2}}}^{(12345)}-\frac14\, \esf_{\cT_{5,3}}^{(12345)}  +\frac{2}{5}\( 2 \esf_{\cT_{5,1}}^{(51234)} 
+\esf_{\cT_{5,1}}^{(45123)} -2 \esf_{\cT_{5,1}}^{(12345)} - \esf_{\cT_{5,1}}^{(23451)} \) 
\end{split}
\ee
and used the fact that $\esf_{\cT_{5,1}}^{(12345)}=\esf_{\cT_{5,1}}^{(54321)}$. 
The advantage of \eqref{En5-m0-first2} is that the coefficients $\esf_{\cT_{5,1}}$ appear in the combination that
is subject to the identity \eqref{ident-eT}, which for $n=5$ is \eqref{cycle-T51} evaluated at $\bfx=0$.
This implies that all $\esf_{\cT_{5,1}}$ can be replaced by $e_{\cT_{5,1}}=1/5$, in agreement with the value 
given in the table of appendix \ref{ap-coef}.

Thus, it remains to analyze only the last term in \eqref{En5-m0-first2}. As a first step, we apply the symmetry 
of $\frK_{\cT_{5,2}}^{(12345)}$ under permutation $\sigma_{45}$ to rewrite this term as
\be
\label{def-hesf}
\sum_{\sigma \in S_5}
\sigma\Bigl[\frK_{\cT_{5,2}}\,\check\esf_{\cT_{5,2}}\Bigr]
= -\frac{1}{20}\sum_{\sigma \in S_5}\frK_{\cT_{5,2}}^\sigma+\cC_{5,2},
\qquad
\cC_{5,2}=\sum_{\sigma \in S_5}
\sigma\Bigl[\frK_{\cT_{5,2}}\,\hesf_{\cT_{5,2}}\Bigr]
\ee
where
\be
\hesf_{\cT_{5,2}} =\hf\( \check\esf_{\cT_{5,2}}^{(12354)}+\check\esf_{\cT_{5,2}}^{(12354)}\)+\frac{1}{20}\, .
\label{def-hesf}
\ee
Here we subtracted $e_{\cT_{5,2}}-\frac14\, e_{\cT_{5,3}}=-\frac{1}{20}$, which is the rational number 
expected to replace $\check\esf_{\cT_{5,2}}$ according to appendix \ref{ap-coef}, so that now we need to prove 
that the contribution $\delta_{12345}\cC_{5,2}$ actually vanishes.
This would complete the proof of the Conjecture at $n=5$.

Let us return to the identity \eqref{id-frK-StarFork}. From the symmetry of the tree $\cT_{5,3}$ it follows that
the first term is invariant under any permutation of labels $(1245)$. 
This fact can be used to derive relations that involve only $\cT_{5,2}$:
\be
\begin{split}
	\frK_{\cT_{5,2}}^{(12345)}+\frK_{\cT_{5,2}}^{(14352)}+\frK_{\cT_{5,2}}^{(15324)}
	\stackrel{\delta}{=}
	&\,\frK_{\cT_{5,2}}^{(21345)}+\frK_{\cT_{5,2}}^{(24351)}+\frK_{\cT_{5,2}}^{(25314)}
	\\
	\stackrel{\delta}{=}
	&\,\frK_{\cT_{5,2}}^{(42315)}+\frK_{\cT_{5,2}}^{(41352)}+\frK_{\cT_{5,2}}^{(45321)}
	\\
	\stackrel{\delta}{=}
	&\,\frK_{\cT_{5,2}}^{(52314)}+\frK_{\cT_{5,2}}^{(51342)}+\frK_{\cT_{5,2}}^{(54321)}.
\end{split}
\label{id-frK-ForkFork}
\ee
Then we split the coefficient of $\frK_{\cT_{5,2}}$ as $1=\frac34+\frac{1}{12}+\frac{1}{12}+\frac{1}{12}$,
leave the first term intact and apply the three relations \eqref{id-frK-ForkFork} to the last three terms.
Substituting the resulting expression in $\cC_{5,2}$ and, as before, 
shifting the sum over permutations so that to bring the labels of each factor $\frK_{\cT_{5,2}}$ to
the standard form (12345), one obtains
\be
\cC_{5,2}\stackrel{\delta}{=} 
\sum_{\sigma \in S_5}\sigma\Bigl[
\frK_{\cT_{5,2}}\,\tesf_{\cT_{5,2}}\Bigr],
\label{tesf-52}
\ee
where 
\be
\label{def-tesf}
\begin{split}
\tesf_{\cT_{5,2}} =&\,  \frac34\, \hesf_{\cT_{5,2}}^{(12345)}
-\frac14\(\hesf_{\cT_{5,2}}^{(14325)}+\hesf_{\cT_{5,2}}^{(15324)}\) 
+ \frac{1}{12} \Biggl(\sum_{\tau \in C_3(124)}\tau\[\hesf_{\cT_{5,2}}^{(51324)}\]
\\
&\qquad
+\sum_{\tau \in C_3(145)}\tau\[\hesf_{\cT_{5,2}}^{(21345)} \]
+\sum_{\tau \in C_3(125)}\tau\[\hesf_{\cT_{5,2}}^{(41325)}\]\Biggr),
\end{split}
\ee
where $C_3(ijk)$ is the group of cyclic permutations of the labels $i,j,k$.
Although the new coefficients $\tesf_{\cT_{5,2}}$ have complicated expressions in terms of the original ones 
$\esf_{\cT_{5,k}}$, they turn out to satisfy a very simple identity, 
which can be shown to be a linear combination of the identities \eqref{ident-Phi5}, \eqref{cT5-PhiE-ident123} at $\bfx=0$,
\be
\sum_{\tau\in C_5} \tau\circ(1+\iota)\[\tesf_{\cT_{5,2}}\]=0,
\label{OrbitAvgpm}
\ee 
where $\iota$ is the inversion of labelling introduced in \$\ref{subsec-identPhi}.

Unfortunately, the identity \eqref{OrbitAvgpm} is not sufficient to show the vanishing of $\cC_{5,2}$.
Nevertheless, it suggests a strategy to do this. 
Let us denote $\cO_A$, $A=1,\dots,12$, the orbits of the group $C_5\cup \iota$ inside $S_5$ 
(defined by its action from the {\it right}).
Let $\sigma_A$ be a representative of $\cO_A$, which we take to be the permutation labelled by 
$(i_1i_2 3 i_4i_5)$ with 1 either on the first or second position. We order the orbits according to 
the natural ordering of the representatives (see \eqref{matrix}). It is clear also that $\cO_A$
splits into two orbits of $C_5$: $\cO^+_A=\{\sigma_A \circ\tau, \ \tau\in C_5\}$ and
$\cO^-_A=\cO^+_A\circ\iota$.\footnote{Note that for $\tau\in C_5$ one has $\iota\circ\tau=\tau^{-1}\circ\iota$.} 
With these notations, the identity \eqref{OrbitAvgpm} implies 12 equalities
\be
\ssf_A^+=-\ssf_A^-,
\qquad
\ssf_A^\pm= \sum_{\tau\in \cO^\pm_A} \tesf_{\cT_{5,2}}^\tau.
\label{OrbitAvgpm-s}
\ee

Next, we note that, upon summing over cyclic permutations
of the relation \eqref{id-frK-LinFork}, all terms $\frK_{\cT_{5,1}}$ vanish and one gets a relation
very similar to \eqref{OrbitAvgpm}:
\be
\sum_{\tau \in C_5} \tau\circ(1-\iota)\[\frK_{\cT_{5,2}}\]\stackrel{\delta}{=} 0.
\label{id-frK-Orbits}
\ee
Of course, we can also write it as in \eqref{OrbitAvgpm-s}, but instead we will first use the symmetry 
of $\frK_{\cT_{5,2}}$ under $\sigma_{45}$. Then we get the following 12 relations
\be
\cK_A\equiv \ksf_A^+-\ksf_A^-
\stackrel{\delta}{=}  0,
\qquad 
\ksf_A^\pm= \hf\sum_{\tau\in \cO^\pm_A}\tau\circ(1+\sigma_{45}) \[\frK_{\cT_{5,2}}\].
\label{idOrbits-4545}
\ee
Let us subtract them all from \eqref{tesf-52} after multiplication by some coefficients $c_A$
to be determined later.
In other words, we consider
\be 
\cC_{2,5} \stackrel{\delta}{=} \sum_{\sigma \in S_5}\sigma\Bigl[
\frK_{\cT_{5,2}}\,\tesf_{\cT_{5,2}}\Bigr] 
-\sum_{A=1}^{12} c_A\cK_A=\sum_{\sigma \in S_5}\sigma\Bigl[
\frK_{\cT_{5,2}}\,\resf_{\cT_{5,2}}\Bigr] ,
\ee 
where at the last step we introduced new coefficients $\resf_{\cT_{5,2}}$.
One can realize that they can be written as 
\be
\resf_{\cT_{5,2}}^\sigma = \tesf_{\cT_{5,2}}^\sigma 
- \frac{1}{2} \(\epsilon(\sigma) c_{A(\sigma)} + \eps(\sigma\circ\sigma_{45}) c_{A(\sigma\circ\sigma_{45})}\),
\label{def-resf}
\ee
where $A(\sigma)$ is the label of the orbit to which $\sigma$ belongs and $\eps(\sigma)=\pm 1$ if it is $\cO^\pm_A$,
respectively. 

To fix the coefficients $c_A$ we impose the condition that the cyclic averages of $\resf_{\cT_{5,2}}$,
similar to the ones introduced in \eqref{OrbitAvgpm-s}, identically vanish.
From \eqref{def-resf}, it is easy to see that 
\be 
\sum_{\tau\in \cO^\pm_A} \resf_{\cT_{5,2}}^\tau=\ssf_A^\pm \mp \hf\sum_{B=1}^{12}\cM_{AB}c_B,
\qquad
\cM_{AB}=5\delta_{AB} + m_{AB},
\label{newsum-esf}
\ee
where 
\be 
m_{AB}=\eps\((\cO^+_A\circ\sigma_{45})\cap\cO_B \)
\ee 
and the $\eps$-function is defined to be 0 if the intersection is empty.\footnote{It is well defined 
	because it is impossible to have non-empty intersection with both $\cO^\pm_B$ at the same time.}
With the ordering of orbits specified above, the matrix $m_{AB}$ is given by
\setcounter{MaxMatrixCols}{20}
\be{\footnotesize
\begin{tabular}{cc}
	$A$ \qquad\quad\ \  $\sigma_A$ \ \ \\ \cline{1-1}
\begin{tabular}{ccc}
	1 &: & (12345) \\ 
	2 &: & (12354) \\ 
	3 &: & (14325) \\ 
	4 &: & (14352) \\ 
	5 &: & (15324) \\ 
	6 &: & (15342) \\ 
	7 &: & (21345) \\ 
	8 &: & (21354) \\ 
	9 &: & (41325) \\ 
	10 &: & (41352) \\ 
	11 &: & (51324) \\ 
	12 &: & (51342) \\ 
\end{tabular}
&  \qquad
{\normalsize $m=$}
$\begin{pmatrix}
	0 & 1 & -1 & 0 & 0 & -1 & 1 & 0 & 0 & 0 & 1 & 0 \\
	1 & 0 & 0 & -1 & -1 & 0 & 0 & 1 & 1 & 0 & 0 & 0 \\
	-1 & 0 & 0 & 1 & -1 & 0 & 0 & 0 & 1 & 0 & 0 & 1 \\
	0 & -1 & 1 & 0 & 0 & -1 & 1 & 0 & 0 & 1 & 0 & 0 \\
	0 & -1 & -1 & 0 & 0 & 1 & 0 & 0 & 0 & 1 & 1 & 0 \\
	-1 & 0 & 0 & -1 & 1 & 0 & 0 & 1 & 0 & 0 & 0 & 1 \\
	1 & 0 & 0 & 1 & 0 & 0 & 0 & 1 & -1 & 0 & 0 & 1 \\
	0 & 1 & 0 & 0 & 0 & 1 & 1 & 0 & 0 & 1 & -1 & 0 \\
	0 & 1 & 1 & 0 & 0 & 0 & -1 & 0 & 0 & 1 & 1 & 0 \\
	0 & 0 & 0 & 1 & 1 & 0 & 0 & 1 & 1 & 0 & 0 & -1 \\
	1 & 0 & 0 & 0 & 1 & 0 & 0 & -1 & 1 & 0 & 0 & 1 \\
	0 & 0 & 1 & 0 & 0 & 1 & 1 & 0 & 0 & -1 & 1 & 0
\end{pmatrix}$,
\end{tabular}}
\label{matrix}
\ee
Then one can show that $\cM$ is an invertible matrix with
\be
\cM^{-1}_{AB}=\frac{1}{20}\(5\delta_{AB} - m_{AB}\).
\ee
Therefore, requiring the vanishing of \eqref{newsum-esf} results in
\be 
c_A= \pm \frac{1}{10}\sum_{B=1}^{12}\( 5\delta_{AB} - m_{AB}\) \ssf^\pm_B\, ,
\ee
The two solutions are mutually consistent due to the identities \eqref{OrbitAvgpm-s}.
Picking any of them, by means of \eqref{esf5b}, \eqref{def-hesf}, \eqref{def-tesf} and \eqref{def-resf}, 
one gets an explicit expression of $\resf_{\cT_{5,2}}$ in terms of $\esf_{\cT_{5,k}}$.
Then, using Mathematica, we found that the resulting expression is actually a linear combination of 
the identities \eqref{ident-Phi5} and \eqref{cT5-PhiE-ident123} at $\bfx=0$ and thus vanishes, as required by the Conjecture.

Let us remark in the end that this derivation fixes only a linear combination of two coefficients:
$e_{\cT_{5,2}}-\frac14\, e_{\cT_{5,3}}=-\frac{1}{20}$. 
The ambiguity comes from the fact that, in the presence of $\delta_{12345}$, 
the factors $\frK_{\cT_{5,k}}$ are not independent and satisfy various relations.
It is conceivable that the ambiguity is resolved at higher orders. However, the calculations 
at higher $n$ become extremely cumbersome and we refrain from doing them.
What is important is that the values predicted by the general solution \eqref{res-eT} in the main text
are consistent with the above constraint.

\providecommand{\href}[2]{#2}\begingroup\raggedright\endgroup

%\bibliography{combined}
%\bibliographystyle{utphys}

\end{document}